%% file: eptcs-main.tex
\documentclass[submission,copyright,creativecommons]{eptcs}
\providecommand{\event}{MFPS 2021} 
\usepackage{underscore}           

\AtBeginDocument{}

\usepackage[english]{babel}

\usepackage{
  amssymb,
  amsmath,
  amsthm,
  stmaryrd,
  thm-restate,
  mathpartir,
  enumitem,
  thm-restate
} 

\usepackage[mathcal]{euscript}

\usepackage[outline]{contour}

\usepackage{tikz} 
\usetikzlibrary{
  automata,
  positioning,
  arrows,
  shapes.geometric,
  calc,
  cd
}

\tikzstyle{automaton} = [
  ->,
  > = stealth,
  node distance = 2cm,
  every state/.style = {
    thick,
    fill=black!0.0,
    draw=black!0.0
  }
]

\DeclareFontFamily{U}{mathb}{\hyphenchar\font45}
\DeclareFontShape{U}{mathb}{m}{n}{
<-6> mathb5 <6-7> mathb6 <7-8> mathb7
<8-9> mathb8 <9-10> mathb9
<10-12> mathb10 <12-> mathb12
}{}
\DeclareSymbolFont{mathb}{U}{mathb}{m}{n}
\DeclareMathSymbol{\lefttorightarrow}      {3}{mathb}{"FC}
\DeclareMathSymbol{\righttoleftarrow}      {3}{mathb}{"FD}

\title{On Star Expressions and \\Coalgebraic Completeness Theorems}
\author{Todd Schmid
\institute{Department of Computer Science\\
UCL\\
London, UK}
\email{todd.schmid.19@ucl.ac.uk}
\and
Jurriaan Rot
\institute{Department of Computer Science\\
Radboud University\\
Nijmegen, The Netherlands}
\email{jrot@cs.ru.nl}
\and
Alexandra Silva\footnote{Silva’s work was partially supported by ERC grant Autoprobe (grant agreement 101002697) and a Royal society Wolfson fellowship.}
\institute{Department of Computer Science\\
UCL\\
London, UK}
\email{alexandra.silva@ucl.ac.uk}
}
\def\titlerunning{On Star Expressions and Completeness Theorems}
\def\authorrunning{T. Schmid, J. Rot, \& A. Silva}

\begin{document}

\input{symbols.tex}

\maketitle

\begin{abstract}
	\input{abstract.tex}
\end{abstract}

\input{article.tex}

\nocite{*}
\bibliographystyle{eptcs}
\bibliography{citations}

\appendix

\input{appendix.tex}

\end{document}

%% file: symbols.tex

\theoremstyle{plain}
\newtheorem{theorem}{Theorem}[section]
\newtheorem{lemma}{Lemma}[section]
\newtheorem{corollary}{Corollary}[section]
\newtheorem{proposition}{Proposition}[section]
\theoremstyle{definition}
\newtheorem{remark}{Remark}[section]
\newtheorem{definition}{Definition}[section]
\newtheorem{example}{Example}[section]
\newtheorem{conjecture}{Conjecture}[section]

\newcommand	\sem[1]		{\llbracket#1\rrbracket}
\newcommand \trans[1]	{\mathrel{\raisebox{-2pt}{\(\xrightarrow{#1}\)}}}

\newcommand	\filter		{\mathop{@}}

\newcommand	\Cat[1] 	{\operatorname{{\bf #1}}} 
\newcommand \define[1]  {\emph{\textbf{#1}}} 

\newcommand	\N 			{{\mathbb N}} 
\newcommand \SExp       {{\operatorname{SExp}}} 
\newcommand \acro[1] 	{{\sf #1}}

\newcommand \Sets 		{{\Cat{Sets}}} 
\newcommand \Coalg      {{\operatorname{Coalg}}} 
\newcommand \Cov		{{\operatorname{Cov}}} 
\newcommand \Id 		{\operatorname{Id}} 
\newcommand \Pfin 		{\mathcal P_{\omega}}

\newcommand \coalg[1]	{{\mathcal #1}} 
\newcommand \coeq[1]	{{\mathit #1}} 
\newcommand \V 			{{\coalg V}} 
\newcommand \X 			{{\coalg X}} 
\newcommand \Y 			{{\coalg Y}} 
\newcommand \Z 			{{\coalg Z}} 
\newcommand \R  		{{\coalg R}} 
\newcommand \Q  		{{\coalg Q}} 

\newcommand \id 		{\operatorname{id}} 
\newcommand \beh 		{{\operatorname{beh}}} 
\newcommand \loopright  {\mathrel{\lefttorightarrow}}
\newcommand \loopleft   {\mathrel{\righttoleftarrow}}
\newcommand \bisim 		{\mathrel{\raisebox{1pt}{\(\underline{\leftrightarrow}\)}}} 
\newcommand \bo         {\to_{\sf b}}
\newcommand \eo 		{\mathbin{\contourlength{0.02em}\contour{black}{\(\to_{\sf e}\)}}}
\newcommand \diredge 	{\curvearrowright} 
\newcommand \into 		{\mathrel{\hookrightarrow}}
\newcommand \onto 		{\mathrel{\twoheadrightarrow}}

\newcommand \Node 		{{\operatorname{Node}}} 
\newcommand \Leaf 		{{\operatorname{Leaf}}} 

\newcommand \llb 		{\llbracket} 
\newcommand \rrb 		{\rrbracket} 
\newcommand \nin        {\mathbin{\not\in}} 
\newcommand \img 		{{\operatorname{img}}}
\newcommand \dom 		{{\operatorname{dom}}}
\newcommand \op 		{{\text{op}}}

%% file: abstract.tex
An open problem posed by Milner asks for a proof that a certain axiomatisation, which Milner showed is sound with respect to bisimilarity for regular expressions, is also complete. 
One of the main difficulties of the problem is the lack of a full \emph{Kleene theorem}, since there are automata that can not be specified, up to bisimilarity, by an expression.
Grabmayer and Fokkink (2020) characterise those automata that can be expressed by regular expressions without the constant 1, and use this characterisation to give a positive answer to Milner's question for this subset of expressions.
In this paper, we analyse Grabmayer and Fokkink's proof of completeness from the perspective of universal coalgebra, and thereby give an abstract account of their proof method.
We then compare this proof method to another approach to completeness proofs from coalgebraic language theory.
This culminates in two abstract proof methods for completeness, what we call the \emph{local} and \emph{global} approaches, and a description of when one method can be used in place of the other.


%% file: article.tex

\section{Introduction}\label{sec:introduction}
	In 1984, Robin Milner gave a non-standard operational interpretation of regular expressions~\cite{milner1984complete}, viewing them as nondeterministic processes rather than regular languages.
	Comparing them for bisimilarity rather than language equivalence affects the semantics in two key ways.
	First, there are finite nondeterministic processes that do not behave like any regular expression 
	up to bisimilarity (the problem of characterising those that do was solved first in~\cite{BaetenCG07}).
	This draws a stark contrast with the language semantics of regular expressions, where Kleene's theorem gives a one-to-one correspondence between finite automata and expressions.
	Second, there are axioms in Salomaa's complete axiomatisation of the algebra of regular expressions~\cite{salomaa1966two} that are unsound in the process interpretation.
	Milner offers a modified version of Salomaa's axioms and shows that they are sound with respect to bisimilarity.
	Completeness is left as an open problem in \cite{milner1984complete}, a full solution to which has yet to appear in the literature.

	Several partial solutions to Milner's problem are known.
	For instance, by omitting the constants \(0\) and \(1\) from the language and replacing the Kleene star with its binary version,\footnote{In fact, Kleene's original star operation was binary. However, the binary star operation seems to have fallen out of fashion by the time \cite{salomaa1966two} was written.} interpreted as iteration on its first argument before proceeding with the second, one obtains the calculus introduced in \cite{bergstrabethkeponse1994process}.
	The corresponding axiomatisation was shown to be complete with respect to bisimilarity in \cite{fokkinkzantema1994basic,fokkinkzantema1997termination}.
	Later, soundness and completeness were shown for the fragments including either (or both) of \(0\) and \(1\), but with a perpetual loop operator \((-)^\omega\) in place of the star~\cite{fokkink1997perpetual}.
	These partial solutions led up to the recent paper of Grabmayer and Fokkink~\cite{grabmayerfokkink2020complete}, which solves the completeness problem for the fragment of \(1\)-free regular expressions, and so subsumes much of the previous work on the problem.
	However, what specifically interests us in their work is that it perfectly illustrates a subtle difficulty in completeness proofs. 
	
	
	Grabmayer and Fokkink's approach consists of four key parts, as is the case for many related completeness proofs that go through automata.
	The first is the production of models from expressions through the operational semantics.
	The second is a sort of inverse to the first, a notion of \emph{solution to a model} in the class 
	of expressions. 
	The third is the identification of a distinguished class of models that includes the semantics of the expressions, every member of which admits a \emph{unique} solution modulo the axioms.
	This gives a one-to-one correspondence between models in the distinguished class and expressions modulo the axioms.
	The fourth is the ability to combine or reduce and compare models (without leaving the distinguished class), as well as their solutions.
	
	The last part is subtler than the first three.
	In a classical proof such as~\cite{salomaa1966two}, but also in more recent coalgebraic formulations (e.g.,~\cite{silva2010kleene,jacobs2006bialgebraic}), the distinguished class typically consists of all finite (or locally finite) automata, and comparing automata consists of finding a bisimulation between them.
	Bisimulations between finite automata are trivially finite, so the fourth step is rarely worth mentioning in this situation.
	Here, the highest hurdle to clear seems to be the issue of proving that solutions are unique.
	This is in stark contrast with the setting of Grabmayer and Fokkink's paper~\cite{grabmayerfokkink2020complete}, where the class of models they consider are the so-called \emph{\acro{LLEE}-charts}. 
	The necessity of identifying such a non-trivial class is caused by the above-mentioned issue that, up to bisimilarity, not all finite processes are characterised by regular expressions. 
	For \acro{LLEE}-charts, uniqueness of solutions is not a triviality, but also does not warrant a proof in the main body of \cite{grabmayerfokkink2020complete}.\footnote{It appears in the extended version \cite{grabmayerfokkink2020extended}.} 
	In comparison, a great amount of ingenuity is involved in establishing the fourth of the moving parts mentioned above: The ability to reduce equivalent \acro{LLEE}-charts to a common \acro{LLEE}-chart. 
	
	Grabmayer and Fokkink's solution to this problem is highly innovative and technical, and makes use of new tools carefully crafted for proving the compositionality result mentioned above.
	The abstract view we present here is no replacement for the detailed combinatorial arguments found in \cite{grabmayerfokkink2020complete}.
	Instead, the intent of the present paper is to unpack its contents by situating them in the context of universal coalgebra.
	Universal coalgebra is a well-established general framework for state-based systems, subsuming constructs like bisimilarity and behaviour~\cite{rutten2000universal,jacobs2016introduction}.
	We give a coalgebraic spin on some of the results of \cite{grabmayerfokkink2020complete}, strengthening some while simplifying the proofs of others:
	\begin{itemize}
		\item We show that solutions to automata are in one-to-one correspondence with coalgebra homomorphisms into the expressions modulo the axioms.
		\item We elucidate the four moving parts of Grabmayer and Fokkink's completeness proof mentioned above and prove that they are sufficient in a general coalgebraic setting.
		\item We generalise the \emph{connect-through-to} operation from \cite{grabmayerfokkink2020complete} to a purely coalgebraic construction. 
		We coin the term \emph{rerouting} for this construction, and show that a prevalence of reroutings can be used to establish the fourth moving part of completeness proofs.
		\item Finally, we give a general account of a related approach to completeness proofs found in \cite{jacobs2006bialgebraic,silva2010kleene,milius2010streamcircuits,bonsanguemiliussilva2013sound} and show how the method used by Grabmayer and Fokkink can be restructured to fit this mould.
	\end{itemize}
	Overall, we use the structure of the completeness proof in \cite{grabmayerfokkink2020complete} as a case study in completeness proof methods from coalgebra that do not rely on a one-to-one correspondence between expressions and all finite automata.	
	This culminates in two abstract proof methods for completeness, what we call the \emph{local} and \emph{global} approaches, and a description of those situations in which the latter method can be used in place of the former. 
%

	

	
	The paper is organized as follows:
	In \autoref{sec:coalgebras and 1-free star expressions}, we introduce the \(1\)-free fragment of regular expressions in parallel with its coalgebraic aspects.
	In \autoref{sec:coalgebras and completeness}, we discuss the four moving parts of Grabmayer and Fokkink's completeness proof and show that they are sufficient in a general coalgebraic setting. 
	In \autoref{sec:layered loop existence and elimination}, we give an alternative description of \acro{LLEE}-charts and show how Grabmayer and Fokkink's technique for reducing \acro{LLEE}-charts can be strengthened.
	It is in this section that we generalize their \emph{connect-through-to} operation.
	Lastly, in \autoref{sec:a global approach}, we give a general account of a related approach to completeness proofs, found in \cite{jacobs2006bialgebraic,silva2010kleene,milius2010streamcircuits,bonsanguemiliussilva2013sound}, and show how the method used by Grabmayer and Fokkink can be restructured to fit this mould.



	
\section{Coalgebras and 1-free Star Expressions}\label{sec:coalgebras and 1-free star expressions}
	For a fixed finite set \(A\) of \emph{atomic actions}, the set of \emph{\(1\)-free star expressions}, or \emph{star expressions} for short, is generated by the BNF grammar
	\[
		\SExp \ni e,f ::= a \in A \mid 0 \mid e + f \mid ef \mid e * f
	\]
	The expression \(e*f\) denotes the regular expression \(e^* f\) from \cite{milner1984complete}, but we write \(*\) as an infix to emphasize that it is a binary operation in this formalism, as in Kleene's seminal paper~\cite{kleene1951representation}.

	Operationally, each star expression specifies a labelled transition system with outputs in \(2^A\), called its \emph{chart}.
	Following~\cite{grabmayerfokkink2020complete}, a \emph{chart} consists of a set of \emph{states} \(X\), a \emph{transition} relation \(\trans{(-)}\ \subseteq X \times A \times X\), an \emph{output} relation \(\Rightarrow\ \subseteq X \times A\), and a \emph{start} state \(x \in X\) from which every other state is reachable via a path of finite length.
	We impose the additional assumption that charts are \emph{finitely branching}, ie. for any \(x \in X\) and \(a \in A\), \(x \trans{a} y\) for finitely many \(y \in X\).
	
	Where \(\Pfin(X) = \{U \subseteq X \mid |U| < \omega\}\), a transition relation is equivalent to a function \(\partial: X \to \Pfin(X)^A\), an output relation is equivalent to a function \(o : X \to 2^A\), and they can be given together by a function
	\[
		\langle o, \partial\rangle : X \to 2^A \times \Pfin(X)^A.
	\]
	Since \(\langle o, \partial\rangle\) says nothing about a start state, we call a pair \((X, \langle o, \partial\rangle)\) a \emph{prechart}.

	Precharts fit nicely into the framework of universal coalgebra.
	For an endofunctor \(G\) on the category \(\Sets\) of sets and functions, a \emph{\(G\)-coalgebra} is a pair \((X, \delta_\X)\) consisting of a set \(X\) of \emph{states} and a \emph{structure map} \(\delta_\X : X \to GX\).
	Thus, if \(P(X) = 2^A \times \Pfin(X)^A\) and \(P(f : X \to Y)(o, h)(a) = (o(a), f(h(a)))\), precharts are precisely \(P\)-coalgebras.
	Given a prechart \((X, \langle o, \partial\rangle)\), its transition and output relations can be recovered by writing \(x \trans{a} y\) to denote \(y \in \partial(x)(a)\) and \(x \Rightarrow a\) to denote \(o(x)(a) = 1\).

	To obtain a chart from each star expression, Grabmayer and Fokkink begin by giving the set of star expressions \(\SExp\) the structure of a prechart. 
	The transitions of \(\SExp\) are built inductively from the interpretations of expressions as processes: The constant \(0\) is deadlock, \(a \in A\) is the process that performs the action \(a\) and then terminates, \(e + f\) and \(ef\) are alternative and sequential composition respectively, and \(e * f\) iterates \(e\) before executing \(f\).
	Formally, the transitions and outputs of \(\SExp\) are those derivable from the rules in \autoref{fig:prechart_structure_of_sexp}.

	\begin{figure}[!t]
		\begin{mathpar}
			\infer{a \in A}{a \Rightarrow a}
			\and
			\infer{e_i \Rightarrow a}{e_1 + e_2 \Rightarrow a}
			\and
			\infer{e_i \trans{a} f}{e_1 + e_2 \trans{a} f}
			\and
			\infer{e_1 \Rightarrow a}{e_1e_2 \trans{a} e_2}
			\and
			\infer{e_1 \trans{a} f}{e_1e_2 \trans{a} fe_2} 
			\\
			\infer{e_2 \Rightarrow a}{e_1*e_2 \Rightarrow a}
			\and
			\infer{e_2 \trans{a} f}{e_1*e_2\trans{a} f}
			\and
			\infer{e_1 \trans{a} f}{e_1*e_2 \trans{a} f(e_1*e_2)}
			\and
			\infer{e_1 \Rightarrow a}{e_1*e_2 \trans{a} e_1*e_2}
		\end{mathpar}
		\caption{The prechart \((\SExp, \langle o_\SExp, \partial_\SExp\rangle)\).}
		\label{fig:prechart_structure_of_sexp}
	\end{figure}

	Given an expression \(e \in \SExp\), the \emph{chart interpretation of \(e\)} is the smallest subset of \(\SExp\) containing \(e\) and closed under the transition and output relations.
	The resulting prechart is denoted \(\langle e \rangle\), and coincides with the smallest \emph{subcoalgebra} of \(\SExp\) containing \(e\), ie. if \(U \subseteq \SExp\) contains \(e\) and \((U, \langle o_U, \partial_U\rangle)\) is a \(P\)-coalgebra such that
	\begin{equation}\label{eq:inclusion homomrphism}
		\begin{tikzcd}
		U \ar[rr, hook, "{\text{in}_U}"] \ar[d, "{\langle o_U, \partial_U\rangle}"'] && \SExp \ar[d, "{\langle o_\SExp, \partial_\SExp\rangle}"] \\
		P(U) \ar[rr, hook, "{P(\text{in}_U)}"] && P(\SExp)
	\end{tikzcd}
	\end{equation}
	commutes, then \(\langle e \rangle \subseteq U\).
	In general, a \emph{chart} is a prechart of the form \(\langle x\rangle\) for some \((X, \langle o,\partial\rangle)\) with \(x \in X\).

	In coalgebraic terminology, \eqref{eq:inclusion homomrphism} states that the set \(U\) carries a \(P\)-coalgebra structure such that the inclusion of \(U\) into \(\SExp\) is a \emph{\(P\)-coalgebra homomorphism}.
	For a general endofunctor \(G\) on \(\Sets\), a \emph{\(G\)-coalgebra homomorphism} from \((X, \delta_X)\) to \((Y, \delta_Y)\) is a map \(h : X \to Y\) such that \(\delta_Y\circ h = G(h) \circ \delta_X\).
	We write \(X \cong Y\) if there is a bijective coalgebra homomorphism \(X\to Y\), and say that \(X\) and \(Y\) are \emph{isomorphic}.
	Homomorphisms coincide with the standard notion of \emph{functional bisimulation}.

	\begin{restatable}{lemma}{restatefunctionalbisimulations}\label{lem:functional bisimulations}
		A function \(h : X \to Y\) between precharts is a coalgebra homomorphism if and only if for any \(x \in X\), \(y' \in Y\), and \(a \in A\), (i) \(x \Rightarrow a\) if and only if \(h(x) \Rightarrow a\), and (ii) \(h(x) \trans{a} y'\) if and only if there is an \(x' \in X\) such that \(h(x') = y'\) and \(x \trans{a} x'\).
	\end{restatable}

	\emph{Bisimulation} can also be captured coalgebraically.
	For a general endofunctor \(G\), a \emph{bisimulation} between two coalgebras \((X, \delta_X)\) and \((Y, \delta_Y)\) is a relation \(R \subseteq X \times Y\) carrying a coalgebra structure \((R, \delta_R)\) such that the projection maps \(\pi_1 : R \to X\) and \(\pi_2 : R \to Y\) are \(G\)-coalgebra homomorphisms.
	It follows from Lemma~\ref{lem:functional bisimulations} that a relation \(R \subseteq X \times Y\) between precharts is a bisimulation if and only if for any \((x,y) \in R\) and \(a \in A\), (i) \(x \Rightarrow a\) if and only if \(y \Rightarrow a\); (ii) if \(x \trans{a} x'\), then there is a \(y' \in Y\) such that \((x',y') \in R\) and \(y \trans{a} y'\); and (iii) if \(y \trans{a} y'\), then there is an \(x' \in X\) such that \(x \trans{a} x'\) and \((x',y') \in R\).
	Conversely, a map \(h\) is a coalgebra homomorphism if and only if its graph \(\operatorname{Gr}(h) = \{(x,h(x))\mid x \in X\}\) is a bisimulation.
	If there is a bisimulation \(R\) relating \(x \in X\) and \(y \in Y\), we say \(x\) and \(y\) are \emph{bisimilar} and write \(x \bisim y\).
	Restricted to a single coalgebra, \({\bisim} \subseteq X \times X\) is a \emph{bisimulation equivalence}, a bisimulation that is also an equivalence relation.
	
	Within \(\SExp\), bisimilarity satisfies a number of intuitive equivalences, keeping in mind the interpretation of star expressions as processes.
	For instance, \(0e \bisim 0\) and \(e + f \bisim f + e\) for any \(e,f \in \SExp\).
	These are captured by two of the axioms suggested by Milner in \cite{milner1984complete}, appearing as (B7) and (B1) in Grabmayer and Fokkink's adaptation of Milner's axioms to star expressions summarised in \autoref{fig:the_axioms_found_in_cite_grabmayerfokkink2020complete}.
	We define \(\equiv\) to be the smallest congruence relation on \(\SExp\) containing the pairs \(e \equiv f\) found in \autoref{fig:the_axioms_found_in_cite_grabmayerfokkink2020complete}.

	\begin{figure}[!t]
		\begin{mathpar}
		\begin{array}{c r c l}
			\text{(B1)}		&		e_1 + e_2 	&\equiv& e_2 + e_1\\
			\text{(B2)}		& e_1 + (e_2 + e_3) &\equiv& (e_1 + e_2) + e_3\\
			\text{(B3)}		&		e_1 + e_1 	&\equiv& e_1\\
			\text{(B4)}		&		(e_1 + e_2)e_3 	&\equiv& e_1e_3 + e_2e_3\\
			\text{(B5)}		&		e_1(e_2e_3) 	&\equiv& (e_1e_2)e_3
		\end{array}
		\and
		\begin{array}{c r c l}
			\text{(B6)}		&		e_1 + 0	&\equiv& e_1\\
			\text{(B7)}		&		0e_1 	&\equiv& 0\\
			\text{(BKS1)}	&  e_1 * e_2&\equiv& e_1 (e_1 * e_2) + e_2\\
			\text{(BKS2)}	&		(e_1 * e_2)e_3 	&\equiv& e_1 * (e_2e_3)\\
			\text{(RSP)}	&		&&\hspace{-4em}\infer{e_3 \equiv e_1 e_3 + e_2}{e_3 \equiv e_1*e_2}
		\end{array}

		\end{mathpar}
		\caption{A sound and complete axiomatisation \cite{grabmayerfokkink2020complete}. Here, \(e_1,e_2,e_3 \in \SExp\).}
		\label{fig:the_axioms_found_in_cite_grabmayerfokkink2020complete}
	\end{figure}

	In general, an equivalence relation \(\equiv\) on the state space of a \(G\)-coalgebra \(E\) is \emph{sound} with respect to bisimilarity if \(\equiv\) is a bisimulation equivalence, and \emph{complete} with respect to bisimilarity if \(e \equiv f\) whenever \(e \bisim f\).
	The following theorem says that the axioms in \autoref{fig:the_axioms_found_in_cite_grabmayerfokkink2020complete} build a sound equivalence relation on \(\SExp\).

	\begin{restatable}{theorem}{restatestrongsoundness}\label{thm:strong soundness}
		The relation \({\equiv} \subseteq \SExp \times \SExp\) is a bisimulation equivalence on \(\SExp\).
	\end{restatable}

	The role that bisimulation equivalences play in coalgebra is analogous to the role that congruences play in algebra.
	The kernel \(\ker(h) = \{(x, x') \in X \times X \mid h(x) = h(x')\}\) of any coalgebra homomorphism \(h\) is a bisimulation equivalence,\footnote{Actually, this is only true if \(G\) preserves weak pullbacks. This is a common assumption, however, and holds for each of the concrete functors we consider here.} and every bisimulation equivalence \(R\) is the kernel of some coalgebra homomorphism \(X \to X/R\)~\cite{rutten2000universal}.
	By Theorem \ref{thm:strong soundness}, the set \(\SExp/{\equiv}\) of star expressions modulo provable equivalence is itself a \(P\)-coalgebra, and the quotient map \([-]_{\equiv} : \SExp \to \SExp/{\equiv}\) is a coalgebra homomorphism.

	\subsection{Linear Systems and Solutions}
	Starting with an expression \(e \in \SExp\), obtaining a prechart \(X\) with a state \(x \in X\) such that \(e \bisim x\) is only a matter of computing \(\langle e \rangle\).
	However, going from a prechart \(X\) and a state \(x \in X\) to an expression \(e \in \SExp\) such that \(e \bisim x\) is more difficult (and in fact, is not always possible).
	The following theorem hints at a method for doing so.

	\begin{restatable}{theorem}{restatefundamentaltheorem}\label{thm:fundamental}
	 	Let \(e \in \SExp\).
	 	Then
	 	\(
			e \equiv \sum\limits_{e \Rightarrow a} a + \sum\limits_{e \trans{a} f} af
		\),
		where \(\sum_{i=1}^n e_i = e_1 + \left(\sum_{i = 2}^n e_i\right)\).\footnote{Here, the generalised sum on the left is well-defined up to the commutativity and associativity of \(+\) assumed in Figure \ref{fig:the_axioms_found_in_cite_grabmayerfokkink2020complete}.}
	\end{restatable} 
	
	Given a finite prechart \((X, \langle o, \partial\rangle)\), its corresponding \emph{linear system} is the set of equations
	\begin{equation}\label{eq:typical linear system}
		x = \sum_{x \Rightarrow a} a + \sum_{x \trans{a} x'} a x'	
	\end{equation}
	indexed by \(X\), where we are thinking of each \(x \in X\) as an indeterminate.
	A \emph{solution} to the linear system associated with \(X\) is a map \(s : X \to \SExp\) such that 
	\begin{equation}\label{eq:solution form1}
		s(x) \equiv \sum_{x \Rightarrow a} a + \sum_{x \trans{a} x'} a~s(x')
	\end{equation}
	for all \(x \in X\). 
	Composing a solution \(s\) with the homomorphic image homomorphism \([-]_{\equiv} : \SExp \to \SExp/{\equiv}\), \eqref{eq:solution form1} becomes the equation
	\begin{equation}\label{eq:solution composed}
		[s(x)]_{\equiv} = \sum_{x \Rightarrow a} a + \sum_{x \trans{a} x'} a~[s(x')]_{\equiv}
	\end{equation}
	It follows from \eqref{eq:solution composed} that \(\{(x, [s(x)]_{\equiv}) \mid x \in X\}\) is a bisimulation between \(X\) and \(\SExp/{\equiv}\).
	Since this is the graph of the map \([-]_{\equiv} \circ s\), if \(s : X \to \SExp\) is a solution, then \([-]_{\equiv} \circ s : X \to \SExp/{\equiv}\) is a coalgebra homomorphism.
	Conversely, if \([-]_{\equiv}\circ s\) is a homomorphism, then \eqref{eq:solution composed} holds.
	As \eqref{eq:solution form1} and \eqref{eq:solution composed} are equivalent, we obtain:

	\begin{restatable}{lemma}{restatesolutionlemma}\label{lem:solution lemma}
		A map \(s : X \to \SExp\) is a solution iff \([-]_{\equiv} \circ s : X \to \SExp/{\equiv}\) is a coalgebra homomorphism.
	\end{restatable} 

	\noindent We often identify solutions with their corresponding homomorphisms into \(\SExp/{\equiv}\).

\section{A Local Approach}\label{sec:coalgebras and completeness}
	In the previous section, we observed that the axioms in \autoref{fig:the_axioms_found_in_cite_grabmayerfokkink2020complete} are sound with respect to bisimilarity, and that \emph{solutions} from \cite{grabmayerfokkink2020complete} coincide with coalgebra homomorphisms into \(\SExp/{\equiv}\).
	In \emph{loc. cit.}, Grabmayer and Fokkink show that the axiomatisation in \autoref{fig:the_axioms_found_in_cite_grabmayerfokkink2020complete} is \emph{complete} with respect to bisimilarity: that \(e \equiv f\) whenever \(e \bisim f\), for any \(e,f \in \SExp\).
	Next, we give an abstract description of Grabmayer and Fokkink's approach to proving soundness and completeness, which we call the \emph{local approach}, and compare it to an approach found in classical automata theory.
	Grabmayer and Fokkink's approach can essentially be organized into four steps.
	\begin{description}
		\item[Step 1] is to show that the provable equivalence relation \(\equiv\) is a bisimulation equivalence.
		This is the content of Theorem \ref{thm:strong soundness} from \autoref{sec:coalgebras and 1-free star expressions}, and establishes soundness.
		\item[Step 2] is to identify a class \(\mathcal C\) of precharts such that for any \(e \in \SExp\), \(\langle e \rangle \in \mathcal C\).
		\item[Step 3] is to show that for any \(X \in \mathcal C\), there is a unique homomorphism \(X \to \SExp/{\equiv}\).
		By Lemma \ref{lem:solution lemma}, homomorphisms into \(\SExp/{\equiv}\) are identifiable with solutions, so this is the same as saying that precharts in \(\mathcal C\) admit unique solutions.
		\item[Step 4] is to show that \(\mathcal C\) is closed under binary coproducts and bisimulation collapses. 
		That is, for any \(X,Y \in \mathcal C\), we find \(X\sqcup Y \in \mathcal C\) and \(X/{\bisim} \in \mathcal C\) as well.
	\end{description}
	It should be noted that Grabmayer and Fokkink never explicitly show that their class \(\mathcal C\) is closed under binary coproducts, due to their focus being on charts, which do not have this property.
	Thus, the four steps above are a coalgebraic rephrasing of their approach that requires the introduction of coproducts.
	However, the coalgebraic analogue of Grabmayer and Fokkink's distinguished class of models is easily seen to be closed under binary coproducts, as we will see in \autoref{sec:layered loop existence and elimination}.

	The four steps above are sufficient for showing soundness and completeness of an axiomatisation of bisimilarity in general.
	In fact, we can even replace step 4 with a weaker version: 
	\begin{description}
		\item[Step 4] is to show that \(\mathcal C\) is \emph{collapsible}, ie. for any \(X,Y \in \mathcal C\) and any \(x \in X\) and \(y \in Y\) such that \(x \bisim y\), there is a \(Z \in \mathcal C\) and a pair of homomorphisms \(p : X \to Z\) and \(q : Y \to Z\) such that \(p(x) = q(y)\).
	\end{description}
	Steps 1-4 constitute the \emph{local approach}, leading to soundness and completeness via the following theorem.

	\begin{restatable}{theorem}{restatecompletenessstrategy}\label{thm:completeness strategy}
		Let \(\equiv\) be a bisimulation equivalence on a fixed \(G\)-coalgebra \(E\), and \(\mathcal C\) be a collapsible class of \(G\)-coalgebras containing \(\langle e\rangle\) for each \(e \in E\).
		If there is exactly one homomorphism \(X \to E/{\equiv}\) for every \(G\)-coalgebra \(X \in \mathcal C\), then \(e\equiv f\) if and only if \(e \bisim f\) for any \(e,f \in E\).
	\end{restatable}

	A class of \(G\)-coalgebras that is closed under binary coproducts and bisimulation collapses is collapsible: If \(X\) and \(Y\) are in the class, and \(x \bisim y\) for some \(x \in X\) and \(y \in Y\), let \(Z = (X \sqcup Y)/\bisim\) and take \(p = [-]_{\bisim} \circ \text{in}_X\) and \(q = [-]_{\bisim} \circ \text{in}_{Y}\).
	Here, \(\text{in}_X : X \into X \sqcup Y\) is the inclusion of \(X\) into the coproduct \(X \sqcup Y\), and similarly for \(\text{in}_Y\), and \([-]_{\bisim} : X\sqcup Y \onto Z\) is the bisimulation collapse homomorphism.
	Because \(x \bisim y\) in \(X \sqcup Y\), \(p(x) = q(y)\), from which collapsibility follows.

	Steps 1 through 3 of the local approach should be familiar to readers acquainted with equational axiomatisations in classical automata theory.
	Some aspects of the soundness and completeness theorems of regular algebra can be seen to trace each of the first three steps above.
	Since two states of a deterministic automaton are bisimilar if and only if they recognize the same language, one could reasonably expect that the local approach to proving soundness and completeness, of any one of the existing axiomatisations of language equivalence for regular expressions, should be successful.

	Kleene proved in \cite{kleene1951representation} that a language is regular if and only if it is recognized by a state in a deterministic finite automaton (or DFA).
	This motivates choosing DFAs as the distinguished class of coalgebras.
	This trivializes step 4, as finiteness is preserved under binary coproducts and bisimulation collapses.
	Thus, the central difficulties surpassed in the first completeness proofs of regular algebra lay in step 3~\cite{salomaa1966two,kozen1991completeness}.
	
	Although all four steps had been taken, neither of the completeness proofs in~\cite{salomaa1966two,kozen1991completeness} conclude with an observation like Theorem \ref{thm:completeness strategy}.
	Instead, bisimulations between DFAs are treated as DFAs, and solutions are pulled back across projection homomorphisms.
	As Grabmayer and Fokkink point out in \cite{grabmayerfokkink2020complete}, this use of bisimulations does not translate to the case of \(1\)-free regular expressions.
	This is due to the fact that the distinguished class \(\mathcal C\), consisting of the precharts for which they could prove the existence and uniqueness of solutions, does not include every bisimulation between precharts in \(\mathcal C\).
	This is where the need for collapsibility becomes apparent.

	Comparing the difficulties in Salomaa's approach with the difficulties in Grabmayer and Fokkink's approach reveals a crucial aspect of discovering soundness and completeness theorems in general: When choosing a distinguished class of models \(\mathcal C\), there is a balance to be kept between the difficulty of finding solutions to models in \(\mathcal C\) and proving their uniqueness on the one hand, and ensuring desirable structural qualities of \(\mathcal C\) on the other.
	Salomaa circumvented the difficulties of steps 2 and 4 by including every finite automaton in his distinguished class, but this made step 3 a difficult problem.
	Grabmayer and Fokkink were able to take step 3 and prove uniqueness of solutions for precharts in their distinguished class with relative ease, but step 4 took great ingenuity.

	



\section{Layered Loop Existence and Elimination}\label{sec:layered loop existence and elimination}
	Grabmayer and Fokkink prove that Milner's axioms are complete with respect to bisimilarity for the \(1\)-free fragment by modelling star expressions with charts.
	They single out a specific class of charts, namely those satisfying their \emph{layered loop existence and elimination property}, or \emph{\acro{LLEE}-property} for short.
	Roughly, a prechart is said to satisfy the \acro{LLEE}-property if there is a labelling of its transitions by natural numbers such that an edge descending into a loop accompanies a descent in natural number labellings, and such that no successful termination can occur mid-loop.
	The existence of such a labelling ensures that loops are never mutually nested, and requires threads to finish every task in a loop before termination.
	Every chart interpretation of a star expression has the \acro{LLEE}-property, and every prechart with the \acro{LLEE}-property admits a unique solution.

	In this section, we discuss a coalgebraic version of Grabmayer and Fokkink's distinguished class of models, the class of so-called \emph{\acro{LLEE}-precharts}, and review the proof of its collapsibility.
	As it so happens, a slight variation of Grabmayer and Fokkink's proof of collapsibility shows something much stronger: That the class of finite \emph{\acro{LLEE}-precharts} is closed under arbitrary homomorphic images.
	The main tool used in the proof of collapsibility is the \emph{connect-through-to} operation, which preserves bisimilarity while it identifies bisimilar states.
	We generalize Grabmayer and Fokkink's \emph{connect-through-to} operation, and show that it can be used to establish closure under homomorphic images in general.


	\subsection{Well-layeredness}
	We give an equivalent but different characterisation of \acro{LLEE}-precharts that makes them easier to describe coalgebraically, and rename the property \emph{well-layeredness}.
	While we recall all of the necessary details, much of what is covered here can be found in more detail in \cite{grabmayerfokkink2020complete}.

	A simple but interesting observation about well-layeredness is that it makes no reference to the action labels of a prechart.
	In other words, well-layeredness is really a property of \emph{transition systems with output} (\emph{transition system}s for short), coalgebras for the endofunctor \(2 \times \Pfin(-)\).
	
	A well-layered transition system is a transition system that carries a particular labelling, called an \emph{entry/body labelling}, that satisfies a few extra conditions.
	Here, an \emph{entry/body labelling} of a transition system \((X, \langle \underline o, \underline \partial \rangle)\) is a coalgebra \((X, \langle \underline o, \underline \partial^\bullet\rangle)\) for the endofunctor \(2 \times \Pfin(\{\mathsf{e},\mathsf{b}\}\times (-))\) such that \(\underline\partial(x) = \pi_2(\underline\partial^\bullet(x))\) for any \(x \in X\).
	We typically denote an entry/body labelling of a transition system \(\underline X\) with \(X^\bullet\).

	To state the extra conditions on the labellings that define well-layeredness, we need some notation.
	Given an entry/body labelling \(X^\bullet = (X, \langle \underline o, \underline\partial^\bullet\rangle)\), the following glyphs are used to denote its various transition types: For any \(x,y \in X\), \(x \Rightarrow\) means \(\underline o(x) = 1\), \(x \eo y\) means  \((\mathsf e, y) \in \underline\partial^\bullet(x)\), and \(x \bo y\) means \((\mathsf b, y) \in \underline\partial^\bullet(x)\).
	Furthermore, \(x \diredge y\) means
	\[
			(\exists v_1,\dots,v_k)~ x \eo v_1 \bo \cdots \bo v_k \bo y, \ x \nin \{v_1, \dots, v_l, y\}
	\]
	and \(y \loopright x\) means
	\[
		(\exists v_1,\dots,v_k)~ x \eo v_1 \bo \cdots \bo v_k \bo x, \ y \in \{v_1, \dots, v_l\}, \ x \nin\{v_1, \dots, v_l\}.
	\]
	Transitions of the form \(x\eo y\) and \(x\bo y\) are called \emph{entry} and \emph{body} transitions, respectively.
	We enclose a relation in \((-)^+\) or \((-)^*\) to denote its transitive or transitive-reflexive closure respectively.

	\begin{definition}\label{def:layering witness}
		A \emph{layering witness} is an entry/body labelling \(X^\bullet\) that is
		\begin{enumerate}
			\item \emph{locally finite}, meaning that \(\langle x\rangle\) is finite for all \(x \in X\);
			\item \emph{flat}, meaning that \(x \eo y\) implies \(\neg (x \bo y)\) for all \(x,y \in X\);
			\item \emph{fully specified}, meaning that for all \(x,y \in X\),
			\begin{itemize}
			 	\item[(a)] \(\neg (x \bo^+ x)\) and
			 	\item[(b)] if \(x \eo y\) for some \(y \neq x\), then \(y \to^+ x\).
			 \end{itemize}
			\item \emph{layered}, meaning that the directed graph \((X, \diredge)\) is acyclic; and
			\item \emph{goto-free}, meaning that \(x \diredge y\) implies \(\neg(y \Rightarrow)\), for all \(x,y\in X\).
		\end{enumerate}	
		A transition system is \emph{well-layered} if it is the underlying transition system of a layering witness. 
	\end{definition}
	Every prechart \((X, \langle o, \partial\rangle)\) also comes with an underlying transition system
	\(\underline X = (X, \langle \underline{o}, \underline{\partial} \rangle)\), given by
	\[
		\underline{o} = \bigvee_{a \in A} o(a)
		\qquad\qquad
		\underline{\partial}(x) = \bigcup_{a \in A} \partial(x)(a)
	\]
	for any \(x \in X\).
	A layering witness \emph{for a prechart} is a layering witness for its underlying transition system, and a prechart is said to be \emph{well-layered} if it has a layering witness.

	\begin{remark}
		Every bisimulation \(R\) between precharts \(X\) and \(Y\) carries an underlying bisimulation \(\underline R\) between the transition systems \(\underline X\) and \(\underline Y\).
		However, not every bisimulation between \(\underline X\) and \(\underline Y\) lifts to a bisimulation between \(X\) and \(Y\): Such relations ignore action labels in general, while bisimulations between precharts do not.
	\end{remark}

	\begin{remark}\label{rem:llee from well-layeredness}
		It can be checked that the underlying transition system of a locally finite prechart \(X\) is well-layered if and only if \(X\) has an \emph{\acro{LLEE}-witness}~\cite{grabmayerfokkink2020complete}.
		To obtain an \acro{LLEE}-witness from a layering witness, replace each \(x\eo y\) with a weighted transition \(x \trans{[|x|_{en}]} y\), where\[
			|x|_{en} = \max\{m \in \N \mid \text{\((\exists x_1,\dots,x_m)~x \diredge x_1 \diredge \cdots \diredge x_m\) s.t. \(x \neq x_i \neq x_j\) for \(i \neq j\)}\}
		\] 
		and each \(x \bo y\) with \(x \trans{[0]} y\). 
		This is a well-defined translation because we have assumed that \(\langle x \rangle\) is finite and \((\langle x \rangle, \diredge)\) is acyclic.
		To obtain a layering witness from an \acro{LLEE}-witness, replace each \(x \trans{a}_{[n]} y\) by \(x \eo y\) if \(n > 0\) and \(y \to^+ x\), or by \(x \bo y\) otherwise.
		This entry/body labelling is flat because every resulting entry transition appears in a minimal cycle, and every minimal cycle contains precisely one entry transition by (W1) and (W2)(b) from \cite{grabmayerfokkink2020complete}.
		Each of the remaining conditions are by construction, or are a direct consequence of the \acro{LLEE}-witness conditions.
		For example, full specification follows from (a) local finiteness and (W1) in \emph{loc. cit.}, and (b) our assumption that \(x \eo y\) implies \(y \to^+ x\) for all \(x,y\)~ (see \cite[Proposition C.1]{schmid2021star}).
	\end{remark}

	By restricting a layering witnesses \(X^\bullet\) to a subcoalgebra \(U\) of \(X\), one obtains a layering witness \(U^\bullet\) for \(U\).
	It follows from this observation and the lemma below that \(\langle e \rangle\) is well-layered for any \(1\)-free star expression \(e\).

	\begin{restatable}{lemma}{restatestarexpressionwelllayered}\label{lem:SExp is well-layered}
		The prechart \(\SExp\) is well-layered.
	\end{restatable}

	This completes Step 2 from \autoref{sec:coalgebras and completeness}: Where \(\mathcal C\) is the set of finite well-layered precharts, we find \(\langle e \rangle \in \mathcal C\) for any \(e \in \SExp\).
	For a concrete example, let \(f = (ab)*(ba)\) and \(e = f*a\), where \(a,b,c \in A\).
	The prechart \(\langle e \rangle\) is depicted below along with a layering witness.
	\begin{equation*}
	\begin{tikzpicture}
			\node[state] (0) {\(e\)};
			\node[state, left=3em of 0] (1) {\(ae\)};
			\node[state, right=3em of 0] (2) {\((bf)e\)};
			\node[above=1em of 0] (3) {\(a\)};
			\draw (0) edge[loop below, below, ->] node{\(c\)} (0);
			\path (0) edge[bend left, below, ->] node{\(b\)} (1);
			\path (1) edge[bend left, above, ->] node{\(a\)} (0);
			\path (0) edge[bend left, above, ->] node{\(a\)} (2);
			\path (2) edge[bend left, below, ->] node{\(b\)} (0);
			\path (0) edge[double, double distance=2pt, -implies]  (3);
		\end{tikzpicture}
		\qquad
		\begin{tikzpicture}
			\node[state] (0) {\(e\)};
			\node[state, left=3em of 0] (1) {\(ae\)};
			\node[state, right=3em of 0] (2) {\((bf)e\)};
			\node[above=1em of 0] (3) {};
			\node[below=2em of 0] (4) {};
			\draw (0) edge[loop below, line width=1pt, ->] (0);
			\path (0) edge[bend left, line width=1pt, ->] (1);
			\path (1) edge[bend left, ->] (0);
			\path (0) edge[bend left, line width=1pt, ->] (2);
			\path (2) edge[bend left, ->] (0);
			\path (0) edge[double, double distance=2pt, -implies]  (3);
		\end{tikzpicture}
	\end{equation*}
	It is important to note that not every well-layered prechart has a unique layering witness.
	The prechart \(\langle (aa)*0\rangle\), for example, has exactly two.\medskip
	
	\subsection{Existence and uniqueness of solutions}
	Steps 1 and 2 consisted of showing that \(\equiv\) is a bisimulation equivalence and \(\langle e \rangle\) is a well-layered prechart for each \(e \in \SExp\).
	To complete step 3 of the local approach, Grabmayer and Fokkink give an explicit description of a solution to a chart \(X = \langle v \rangle\) with layering witness \(X^\bullet\), and show that it is equivalent to any other solution to \(X\). 
	For any \(x \in X\), let 
	\begin{equation}\label{eq:canonical solution}
		s_X(x) \equiv \left(\sum_{x \trans{a}_{\mathsf e} x} a + \sum_{\substack{x\trans{a}_{\mathsf e} y\\x \neq y}} a~t_X(y,x)\right)*
				\left(\sum_{x \Rightarrow a} a + \sum_{x\trans{a}_{\mathsf b} y} a~s_X(y)\right)
	\end{equation}
	where \[
		t_X(x,z) \equiv \left(\sum_{x \trans{a}_{\mathsf e} x} a + \sum_{\substack{x\trans{a}_{\mathsf e} y\\x \neq y}} a~t_X(y,x)\right)*
				\left(\sum_{x \Rightarrow a} a + \sum_{x\trans{a}_{\mathsf b} y} a~t_X(y,z)\right)
	\]
	\vspace{0.1em}\\
	Both functions are well-defined by induction on the pair \((|x|_{en}, |x|_{b})\), where \(|x|_{en}\) is given in Remark \ref{rem:llee from well-layeredness} and \(|x|_{b} = \max\{m \mid (\exists x_1,\dots,x_m)\ x \bo x_1 \bo \cdots \bo x_m\}\), with respect to the lexicographical ordering on \(\N\times \N\).
	
	It is shown in \cite{grabmayerfokkink2020complete} that for any solution \(s : X \to \SExp\), \(s(x) \equiv s_X(x)\) for all \(x \in X\).
	This proves that well-layered \emph{charts} have unique solutions.
	The same result readily extends to the prechart case: If \(X\) is an arbitrary well-layered prechart and \(x \in X\), then \(s_{\langle x \rangle}(x)\) is a well-defined expression, as \(\langle x \rangle\) is a subcoalgebra of \(X\) and is therefore also well-layered.
	By uniqueness of solutions for charts, the map \(s_X : X \to \SExp\) given by \(s_X(x) = s_{\langle x \rangle}(x)\) is a well-defined solution to \(X\).
	Furthermore, since every solution to \(X\) restricts to a solution to \(\langle x \rangle\) for each \(x \in X\), \(s_X\) is the unique solution to \(X\).

	\begin{restatable}{lemma}{restateuniquesolutionslemma}\label{thm:unique solutions lemma}
		If \(X\) is a well-layered prechart, then there is a unique solution \(s_X : X \to \SExp/{\equiv}\) to \(X\). 
	\end{restatable}


\subsection{Reroutings and Closure under homomorphic images}
	The crucial step in Grabmayer and Fokkink's proof is step 4 of the local approach, showing that the bisimulation collapse of a finite well-layered chart is also well-layered.
	This is done in a step-by-step procedure that exhaustively identifies bisimilar states.
	In each step, a specially chosen pair \((w_1, w_2)\) of distinct bisimilar states is reduced to the singleton \(w_2\) by \emph{rerouting} all of \(w_1\)'s incoming transitions to \(w_2\) and then deleting \(w_1\).

	Formally, given a prechart \((X, \langle o, \partial\rangle)\) and a pair \((x_1,x_2)\) of distinct states of \(X\), the \emph{connect-\(x_1\)-through-to-\(x_2\) construction} returns the prechart \(X[x_2/x_1] = (X- \{x_1\}, \partial[x_2/x_1])\), where
	\[
	\partial[x_2/x_1](x)(a) = \begin{cases}
			\{x_2\} \cup (\partial(x)(a)- \{x_1\}) & \text{ if \(x_1 \in \partial(x)(a)\),}\\
			\partial(x)(a) & \text{ otherwise}
		\end{cases}
	\]
	The connect-\(x_1\)-through-to-\(x_2\) operation preserves bisimilarity, in the sense that if \(R\) is a bisimulation equivalence on \(X\), then \(R \cap (X \times X-\{x_1\})\) is a bisimulation between \(X\) and \(X[x_2/x_1]\). 
	This has the following consequence: 
	If the only pairs of distinct states in \(R\) are \((x_1,x_2)\) and \((x_2,x_1)\), then \(R_1 = R \cap (X\times X-\{x_1\})\) is the graph of a homomorphism between \(X\) and \(X[x_2/x_1]\), and consequently \(X[x_2/x_1] \cong X/R\).  
	Otherwise, \(R|_{X-\{x_1\}} = R \cap (X-\{x_1\})^2\) is a bisimulation equivalence containing a pair of distinct states \((x_3,x_4)\). 
	If \((x_3,x_4)\) and \((x_4,x_3)\) are the only such pairs, then \(R_2 = R|_{X-\{x_1\}} \cap (X-\{x_1\}\times X-\{x_1,x_3\})\) is the graph of a homomorphism \(X[x_2/x_1]\to X[x_2/x_1][x_4/x_3]\), and therefore \(R_1 \fatsemi R_2 = R \cap (X\times X-\{x_1,x_3\})\) is the graph of a homomorphism \(X \to X[x_2/x_1][x_4/x_3]\), where \(\fatsemi\) denotes relational composition, and \(X[x_2/x_1][x_4/x_3]\cong X/R\).
	Generally, if \(X\) is finite, then iterating this construction yields the graph \(R_1\fatsemi \cdots \fatsemi R_m\) (for some \(m\)) of the homomorphism \(X \to X/R\) (up to \(\cong\)). 
	Taking \(R = {\bisim}\), the bisimulation collapse of a finite prechart \(X\) can be computed by iterating the connect-through-to operation until no distinct pairs of bisimilar states are left.

	For an arbitrary well-layered prechart \(X\) and a pair of distinct bisimilar states \((x_1,x_2)\), \(X[x_2/x_1]\) may not be well-layered.
	An example discussed in \cite{grabmayerfokkink2020complete} is the connect-through-to construction depicted in \autoref{fig:bad rerouting}, which takes a well-layered chart to a chart that does not admit a layering witness.
	
	\begin{figure}[!t]
	\centering
	\begin{tabular}{c c c}
	\footnotesize{\begin{tikzpicture}
			\node[state] (0) {\(x_2\)};
			\node[state, right=3em of 0] (1) {\(v\)};
			\node[state, right=3em of 1] (2) {\(v'\)};
			\node[state, right=3em of 2] (3) {\(x_1\)};
	
			\path (0) edge[bend left, line width=1pt, ->] (1);
			\path (0) edge[bend right=60, line width=1pt, ->] (2);
			\path (1) edge[bend left, ->] (0);
			\path (1) edge[bend left, line width=1pt, ->] (2);
			\path (2) edge[bend left, ->] (1);
			\path (2) edge[->] (3);
			\path (3) edge[bend right=60, ->] (1);
		\end{tikzpicture}}
	&&
	\footnotesize{\begin{tikzpicture}
			\node[state] (0) {\(x_2\)};
			\node[state, right=3em of 0] (1) {\(v\)};
			\node[state, right=3em of 1] (2) {\(v'\)};
	
			\path (0) edge[bend left, ->] (1);
			\path (0) edge[bend right=60, ->] (2);
			\path (1) edge[bend left, ->] (0);
			\path (1) edge[bend left, ->] (2);
			\path (2) edge[bend left, ->] (1);
			\path (2) edge[bend right=60, ->] (0);
		\end{tikzpicture}}
	\\
	\(X^\bullet\)&&\(X[x_2/x_1]\)
	\end{tabular}
	\caption{A bisimulation rerouting that does not preserve well-nestedness.}
	\label{fig:bad rerouting}
	\end{figure}

	However, if \((x_1,x_2)\) is chosen carefully, then the connect-\(x_1\)-through-to-\(x_2\) operation preserves well-layeredness.
	Where \(X^\bullet\) is a layering witness for \(X\), it is shown in \cite{grabmayerfokkink2020complete} that \(X[w_2/w_1]\) is well-layered for any pair \((w_1,w_2)\) of distinct bisimilar states satisfying one of the following three conditions in \(X^\bullet\):
	\begin{center}
		\begin{tabular}{r l}
			(C1) & \(\neg(w_2 \to^* w_1)\); and if \((\exists x)~x \diredge w_1\), then \(\neg(\exists y)(w_2 \to^* y \Rightarrow)\)\\
			(C2) & \(w_2 \loopright^+ w_1\)\\
			(C3) & \(\neg(w_2 \bo^* w_1)\); and \((\exists x)\) \(w_1 \loopright x\) and \(w_2 \loopright^+ x\) and if \(w_1 \loopright y\), then \(x \loopright y\)
		\end{tabular}
	\end{center}
	As Grabmayer and Fokkink point out in \emph{loc. cit.}, if \(X^\bullet\) is a layering witness for a finite prechart \(X\) such that \(X \mathrel{\not\cong} X/{\bisim}\), then there is a pair \((w_1,w_2)\) of distinct bisimilar states satisfying one of (C1)-(C3) in \(X^\bullet\). 
	A slight variation on their proof yields the following.

	\begin{restatable}{lemma}{restatearbitrarybisimulations}\label{lem:arbitrary bisimulations}
		Let \(X^\bullet\) be a layering witness for \(X\), and \(R\) be a bisimulation equivalence on \(X\).
		If \(R\) is non-trivial, ie. \(X \not{\cong} X/R\), then there is a pair \((w_1,w_2) \in R\) of distinct states satisfying one of (C1)-(C3).
	\end{restatable}

	By iterating the connect-through-to construction on the pairs guaranteed to exist in Lemma \ref{lem:arbitrary bisimulations}, every homomorphic image of a finite well-layered prechart is seen to be well-layered.

	\begin{theorem}\label{thm:closure under homomorphic images}
		Let \(X\) be a finite well-layered prechart, and \(R\) be a bisimulation equivalence on \(X\).
		Then \(X/R\) is a well-layered prechart as well.
	\end{theorem}

	This completes step 4 of Grabmayer and Fokkink's proof that Milner's axioms are complete with respect to bisimilarity for the \(1\)-free fragment of regular expressions.

	\begin{theorem}\label{thm:completeness}
		For any \(e,f \in \SExp\), if \(e \bisim f\), then \(e \equiv f\).
	\end{theorem}

	\begin{proof}
		Let \(\mathcal C\) be the set of finite well-layered precharts.
		Lemma \ref{lem:SExp is well-layered} tells us that \(\langle e \rangle \in \mathcal C\) for any \(e \in \SExp\), and Theorem~\ref{thm:unique solutions lemma} tells us that precharts in \(\mathcal C\) admit unique solutions.
		
		By Theorem \ref{thm:completeness strategy}, it suffices to show that \(\mathcal C\) is collapsible.
		We have already seen that a class of coalgebras closed under binary coproducts and homomorphic images is collapsible, so by Theorem~\ref{thm:closure under homomorphic images} it suffices to show that \(\mathcal C\) is closed under binary coproducts.
		To this end, observe that if \(X^\bullet\) and \(Y^\bullet\) are layering witnesses for \(X\) and \(Y\) respectively, then \(X^\bullet \sqcup Y^\bullet\) is a layering witness for \(X \sqcup Y\).
	\end{proof}

	\subsection{Reroutings, in general}
	Interestingly, the connect-through-to construction can be performed on general \(G\)-coal\-gebras.
	For a fixed prechart \(X\) and a pair of states \(x_1,x_2 \in X\), if \(i : X -\{x_1\} \into X\) is the inclusion map and \(j : X \onto X-\{x_1\}\) is the map identifying \(x_2\) with \(x_1\), then the prechart \(X[x_2/x_1] = (X-\{x_1\}, \langle o, \partial[x_2/x_1]\rangle)\) obtained from the connect-\(x_1\)-through-to-\(x_2\) construction is given precisely by
	\[
		\partial[x_2/x_1](x)(a) = j(\partial(x)(a)) = \Pfin(j) \circ \partial \circ i(x)(a).
	\]
	In other words, the following diagram commutes.
	\[\begin{tikzcd}
			X-\{x_1\} \ar[d, "{\langle o, \partial[x_2/x_1]\rangle}"'] \ar[rrr, hook, "i"] &&& X \ar[d, "{\langle o, \partial\rangle}"] \\
			2 \times \Pfin(X-\{x_1\}) &&& 2 \times \Pfin(X)^A \ar[lll, two heads, "\id_2 \times \Pfin(j)^A"']
		\end{tikzcd}\]
	Notice that \((i,j)\) is a \emph{splitting}, meaning \(j \circ i = \id_{X-\{x_1\}}\).
	In general, given any \(G\)-coalgebra \(X\) and any splitting \((i,j)\) with \(i : U \into X\), we define\[
		X[i,j] = (U, G(j) \circ d \circ i),
	\]
	and call \(X[i,j]\) the \emph{rerouting by \((i,j)\)} of \(X\).
	As is the case for the connect-through-to operation, reroutings that identify bisimilar states preserve bisimilarity.


	\begin{restatable}{lemma}{restatebisimulationreroutinglemma}\label{lem:bisimulation rerouting lemma}
		Let \(R\) be a bisimulation equivalence on a prechart \(X\), and \((i,j)\) be a splitting with \(i : U \into X\) and \(\ker(j) \subseteq R\).
		Then \(Q = R \cap (X \times U)\) is a bisimulation.
	\end{restatable}

	A rerouting \(X[i,j]\) is called an \emph{\(R\)-rerouting} if \(R\) is a bisimulation and \(\ker(j) \subseteq R\).
	In case \(R = {\bisim}\), we will use the phrase \emph{bisimulation rerouting} instead.

	A common assumption in universal coalgebra is that the endofunctor \(G\) under consideration preserves weak pullbacks~\cite{rutten2000universal,gumm1998functorsforcoalgebras}.
	This property is sufficient for ensuring that the relational composition of two bisimulations is again a bisimulation.
	In general, if \(R\) is an equivalence on \(X\), and \(Z \subseteq Y \subseteq X\), then \(R \cap (X \times Z) = (R \cap (X \times Y)) \fatsemi (R \cap (Y \times Z))\) and \(R \cap (Y\times Y)\) is an equivalence relation.
	Thus, by iterating Lemma \ref{lem:bisimulation rerouting lemma}, we obtain the following generalisation of Theorem \ref{thm:closure under homomorphic images}.

	\begin{restatable}{theorem}{restatefinitarybisimulationreroutingtheorem}\label{thm:finitary bisimulation rerouting theorem}
		Let \(G\) be an endofunctor that weakly preserves pullbacks, and \(\mathcal C\) be a class of finite \(G\)-coalgebras closed under isomorphism.
		Then the following two statements hold.
		\begin{enumerate}[itemsep=0pt]
			\item If for any \(X \in \mathcal C\) and any nontrivial bisimulation equivalence \(R \subseteq X \times X\) there is a nontrivial \(R\)-rerouting \(U\) of \(X\) such that \(U \in \mathcal C\), then \(\mathcal C\) is closed under homomorphic images.
			\item If for any \(X \in \mathcal C\) such that \(X \mathrel{\not\cong} X/{\bisim}\) there is a nontrivial bisimulation rerouting \(U\) of \(X\) such that \(U \in \mathcal C\), then \(\mathcal C\) is closed under bisimulation collapses.
		\end{enumerate}
	\end{restatable}

	As closure under bisimulation collapses is often enough to establish collapsibility, Theorem \ref{thm:finitary bisimulation rerouting theorem} tells us that establishing an abundance of reroutings in the distinguished class can be a crucial step towards completeness.


\section{A Global Approach}\label{sec:a global approach}
	We now discuss a different approach to proving soundness and completeness theorems in process algebra, which we call the \emph{global approach}, and show how the soundness and completeness theorems of \cite{grabmayerfokkink2020complete} fit in
	this setting.
	Fitting Grabmayer and Fokkink's proof into the mould of the global approach involves expanding the class of finite well-layered precharts to a much larger class that is closed under homomorphic images.
	We further show how the same remoulding technique can remould many local approach proofs into global ones. 

	The global approach originates in coalgebraic automata theory~\cite{jacobs2006bialgebraic,silva2010kleene,silvabonsanguerutten2010nondeterministic,milius2010streamcircuits,bonsanguemiliussilva2013sound}.
	Its main goal is to show that the expression language modulo provable equivalence is isomorphic to a subcoalgebra of a \emph{final} coalgebra.
	For example, in \cite{jacobs2006bialgebraic}, Jacobs proves that the Kleene algebra axioms (see \cite{kozen1991completeness,conway2012regular}) are sound and complete with respect to language equivalence by exhibiting a coalgebra isomorphism between the initial Kleene algebra and the algebra of regular languages.
	The coalgebras that appear in Jacobs' paper are a standard choice for deterministic automata, the \(2 \times (-)^A\)-coalgebras.
	This establishes the soundness and completeness of the Kleene Algebra axioms because bisimilarity and language equivalence coincide.
	Silva successfully applies the same method in \cite{silva2010kleene} to a variety of expression languages and axiomatisations parametrized by the functor \(G\), with Jacobs' proof given by the special case \(G = 2 \times (-)^A\).
	Following the same pattern, Milius gives an expression language and axiomatisation of language equivalence for stream circuits in \cite{milius2010streamcircuits}, and generalizes some of the results in \cite{silva2010kleene} to endofunctors on categories other than \(\Sets\).
	Following a similar approach, all three of the above are unified in \cite{bonsanguemiliussilva2013sound}.

	In order to explain precisely how the global approach works, fix a \(G\)-coalgebra \(E\), thought of as an abstract expression language, and let \(\equiv\) be an equivalence relation on \(E\).
	Similar to the local approach, the global approach involves a sequence of four steps: 
	\begin{description}
		\item[Step 1] is showing that \(\equiv\) is a bisimulation equivalence. 
		This establishes soundness.
		\item[Step 2] consists of identifying a class \(\mathcal C\) of \(G\)-coalgebras in which \(E/{\equiv}\) is \emph{weakly final} in \(\mathcal C\), ie. that \(E/{\equiv} \in \mathcal C\) and every \(X \in \mathcal C\) admits a homomorphism \(X \to E/{\equiv}\).
		Again, homomorphisms into \(E/{\equiv}\) play the role of solutions, so it can be said that coalgebras in \(\mathcal C\) admit solutions.
		\item[Step 3] is a proof that \(E/{\equiv}\) is \emph{final} in \(\mathcal C\), ie. every \(X \in \mathcal C\) admits exactly one solution.
		\item[Step 4] consists of showing that \(\mathcal C\) is closed under homomorphic images.
	\end{description}
	These four steps are sufficient for showing the soundness and completeness of the axiomatisations in each of the cases considered in~\cite{jacobs2006bialgebraic,silva2010kleene,silvabonsanguerutten2010nondeterministic,milius2010streamcircuits,bonsanguemiliussilva2013sound} because the functors that are present there satisfy two key properties.
	The first key property is that there is a \(G\)-coalgebra \(Z\) that is \emph{final}, ie. every \(G\)-coalgebra \(X\) admits a unique homomorphism \(!_X : X \to Z\). 
	Following the steps above, this implies that \(E/{\equiv}\) is a subcoalgebra of \(Z\). 

	\begin{lemma}\label{lem:E/= is subfinal}
		Assume that a final \(G\)-coalgebra exists, and call it \(Z\).
		If \(\mathcal C\) is closed under homomorphic images and has a final object \(Y\), then \(!_Y : Y \to Z\) is injective.
	\end{lemma}

	\begin{proof}
		Where \(!_Y : Y \to Z\) is the unique coalgebra homomorphism from \(Y\) into \(Z\), let \(J = {!}_Z(Y)\).
		The image of a coalgebra homomorphism is always a subcoalgebra of the codomain~\cite{rutten2000universal}, so \(J \in \mathcal C\) by closure under homomorphic images.
		Since \(Y\) is final in \(\mathcal C\), \(J\) admits a unique coalgebra homomorphism \(h : J \to Y\).
		Composing, \(h \circ {!}_Y : Y \to Y\) is a homomorphism, so finality of \(Y\) in \(\mathcal C\) tells us that \(h \circ {!}_Y = \id_Y\). 
		As \(!_Y\) has a left inverse, it is injective.
	\end{proof}

	This means that if every \(X \in \mathcal C\) admits a unique solution and \(\mathcal C\) is closed under homomorphic images, then \([e]_{\equiv} = {!}_{E}(e)\) for any \(e \in E\).\footnote{Here, we have identified \(E/\equiv\) with its isomorphic copy in \(Z\).}
	The second key property is preservation of weak pullbacks.

	\begin{lemma}[Rutten~\cite{rutten2000universal}]\label{lem:bisim = final}
		Let \(X\) and \(Y\) be \(G\)-coalgebras, \(x \in X\), and \(y \in Y\).
		Assume that a final \(G\)-coalgebra exists.
		If \(G\) preserves weak pullbacks, then \(x \bisim y\) if and only if \(!_X(x) = {!}_Y(y)\).
	\end{lemma}

	Following steps 1 through 4 above, and assuming that \(G\) has a final coalgebra and preserves weak pullbacks, Lemmas \ref{lem:E/= is subfinal} and \ref{lem:bisim = final} tell us that \([e]_{\equiv} = {!}_E(e) = {!}_E(f) = [f]_{\equiv}\) if and only if \(e \bisim f\), for any \(e,f \in \SExp\).

	\begin{theorem}\label{thm:global completeness}
		Assume \(G\) preserves weak pullbacks, and let \(\equiv\) be a bisimulation equivalence on a \(G\)-coalgebra \(E\).
		Let \(\mathcal C\) be a class of \(G\)-coalgebras that is closed under homomorphic images.
		If \(E/{\equiv}\) is a final object in \(\mathcal C\), then \(e \equiv f\) if and only if \(e \bisim f\) for any \(e,f \in E\).
	\end{theorem}

	It follows from standard observations about the prechart functor \(P\) that there is a final \(P\)-coalgebra~\cite{rutten2000universal}\footnote{Namely, that it is \emph{bounded}.} and that \(P\) preserves weak pullbacks~\cite{gumm1999elements}.
	This suggests the possibility that the global approach can be taken to proving Theorem \ref{thm:strong soundness} and Theorem \ref{thm:completeness}.
	This is indeed the case, although the class of finite well-layered precharts needs to be extended so as to include \(\SExp/{\equiv}\).

	\subsection{A global approach to the 1-free fragment} 
	Returning to the \(1\)-free fragment of regular expressions, we have already seen that the class of well-layered precharts has \(\SExp\) as a member.
	It is likely that \(\SExp/{\equiv}\) is also well-layered, but proving this turns out to be unnecessary for our purposes.
	
	In order to have the \emph{global} approach go through for the \(1\)-free fragment, we make a slight change in the distinguished class of precharts from \autoref{sec:layered loop existence and elimination}.
	Let \(\mathcal C_{loc}\) be the class of \emph{locally well-layered} precharts, ie. \(X \in \mathcal C_{loc}\) if and only if \(X\) is locally finite and every finite subcoalgebra of \(X\) is well-layered.
	Using the fact that the finite well-layered precharts are closed under homomorphic images, we obtain the following key lemma.

	\begin{restatable}{lemma}{restateclosureunderarbitraryhomomorphicimages}\label{lem:closure under arbitrary homomorphic images}
		Let \(X\) be locally well-layered and \(q : X \onto Y\) a surjective coalgebra homomorphism.
		Then \(Y\) is locally well-layered as well.
	\end{restatable}

	Every well-layered prechart is locally well-layered, so \(\SExp\) is locally well-layered by Lemma \ref{lem:SExp is well-layered}.
	Since \(\SExp \in \mathcal C_{loc}\) and \(\SExp/{\equiv}\) is the image of \(\SExp\) under the homomorphism \([-]_{\equiv} : \SExp \to \SExp/{\equiv}\), Lemma \ref{lem:closure under arbitrary homomorphic images} tells us that \(\SExp/{\equiv} \in \mathcal C_{loc}\) as well.

	So far, we have taken step 4 and the first half of step 2 from the global approach.
	Interestingly, step 3 and the latter half of step 2 are possible because of Theorem \ref{thm:unique solutions lemma}, the uniqueness-of-solutions theorem for finite precharts.
	To see how this works, let \(X \in \mathcal C_{loc}\).
	By Theorem \ref{thm:unique solutions lemma}, every finite subcoalgebra \(U\) of \(X\) admits a unique solution \(s_U : U \to \SExp/{\equiv}\).
	Since homomorphisms restrict to subcoalgebras, this clearly implies that \(X\) admits at most one solution.
	To see that \(\SExp/{\equiv}\) is final in \(\mathcal C_{loc}\), it suffices to construct a solution to \(X\).

	The unique solution to \(X\) is the map \(s_X : X \to \SExp/{\equiv}\) given by \(
		s_X(x) = s_U(x)
	\)
	for any finite subcoalgebra \(U\) of \(X\) containing \(x\).
	To see that this is well-defined, recall that \(X\) is locally finite, meaning that every state of \(X\) is contained in a finite subcoalgebra of \(X\).
	If \(U\) and \(V\) are finite subcoalgebras of \(X\) with \(x \in U\) and \(x \in V\), then \(U \cap V\) is a finite subcoalgebra of \(X\) containing \(x\).
	We have assumed \(U \cap V\) is well-layered, so by Theorem \ref{thm:unique solutions lemma}, \(U \cap V\) admits a unique solution. 
	Restricting \(s_U\) and \(s_V\) to \(U\cap V\) also obtains a solution, so it must be that \(
		s_U(x) = s_{U\cap V}(x) = s_V(x)
	\).
	To see that \(s\) is indeed a solution, observe that a map \(h : X \to Y\) between locally finite coalgebras is a coalgebra homomorphism if \(h|_U : U \to Y\) is a coalgebra homomorphism for any finite subcoalgebra \(U\) of \(X\).
	Since the latter statement is true of \(s\) by definition, \(s\) is a solution to \(X\).
	This establishes the lemma below.

	\begin{lemma}\label{lem:SExp is final}
		Let \(\mathcal C_{loc}\) be the class of locally well-layered precharts.
		Then \(\SExp/{\equiv}\) is a final object in the class \(\mathcal C_{loc}\).
	\end{lemma}

	Together, Theorems~\ref{thm:strong soundness},~\ref{thm:unique solutions lemma}, and Lemmas~\ref{lem:closure under arbitrary homomorphic images},~\ref{lem:SExp is final} constitute steps 1 through 4 of the global approach to proving soundness and completeness of Milner's axioms for the \(1\)-free fragment of regular expressions modulo bisimulation, thus providing an alternative proof of \ref{thm:completeness}.

	\subsection{From local to global}
	As Lemma~\ref{lem:SExp is final} illustrates, there are instances in which a completeness proof taking the global approach can be obtained from the four steps in the local approach.
	This is particularly the case when the distinguished class of coalgebras is closed under binary coproducts and homomorphic images, like the well-layered precharts.
	Where \(E\) is a locally finite \(G\)-coalgebra and \(\equiv\) is a bisimulation equivalence on \(E\), assume that in the four steps of the local approach we have obtained a class \(\mathcal C\) of finite \(G\)-coalgebras such that
	\begin{itemize}
		\item[(a)] each \(X \in \mathcal C\) admits a unique homomorphism into \(E/\equiv\), 
		\item[(b)] \(\langle e \rangle \in \mathcal C\) for any \(e \in E\), and 
		\item[(c)] \(\mathcal C\) is closed under binary coproducts and homomorphic images.
	\end{itemize}
	Then the class \(\mathcal C_{loc}\) of \emph{locally \(\mathcal C\) coalgebras}, locally finite coalgebras \(X\) such that every finite subcoalgebra of \(X\) is in \(\mathcal C\), satisfies the necessary conditions for steps 2 through 4 of the global approach.

	Going through the same motions as in the prechart case, for any \(X \in \mathcal C_{loc}\) the unique solution \(s_X : X \to E/{\equiv}\) is defined locally.
	If \(x \in X\) and \(U\) is a finite subcoalgebra of \(X\) containing \(x\), then \(s_X(x) = s_U(x)\), where \(s_U\) is the unique solution to \(U\).
	Furthermore, if \(h : X \onto Y\) is a surjective coalgebra homomorphism and \(X \in \mathcal C_{loc}\), then for any finite subcoalgebra \(U\) of \(Y\), \(U = h(V)\) for some finite subcoalgebra \(V\) of \(X\).
	By closure under homomorphic images, \(U \in \mathcal C_{loc}\), and by extension \(Y \in \mathcal C_{loc}\) as well.
	Lastly, \(E \in \mathcal C_{loc}\) by definition, so \(E/{\equiv} \in \mathcal C_{loc}\) by closure under homomorphic images.
	The following theorem obtains a global approach-style proof of completeness from the four steps of the local approach when \(\mathcal C\) is closed under coproducts and homomorphic images.

	\begin{restatable}{theorem}{restatelocaltoglobal}\label{thm:local to global}
		Let \(\mathcal C\) be a class of finite \(G\)-coalgebras satisfying (a)-(c) above.
		Then \(\mathcal C_{loc}\) is closed under homomorphic images, and \(E/{\equiv}\) is a final object of \(\mathcal C_{loc}\).
	\end{restatable}

	On the other hand, not every global approach-style completeness proof gives rise to a local one with such immediacy.
	For example, few of the distinguished classes of coalgebras found in the global approach-style proofs in \cite{silva2010kleene} include the DFA interpretation of every expression in the language (each such DFA fails to be locally finite).
	

\section{Discussion and Future Work}
	In this paper, we explore a coalgebraic take on Grabmayer and Fokkink's approach, what we call the local approach, to proving the completeness of Milner's axiomatisation of the \(1\)-free star expressions modulo bisimilarity~\cite{grabmayerfokkink2020complete}.
	We use the insights gained from our exploration to give a general version of their method in \autoref{sec:coalgebras and completeness} that can be applied in other contexts.
	We do the same for a different proof method in \autoref{sec:a global approach}, what we call the global approach, originating in~\cite{jacobs2006bialgebraic,silva2010kleene,bonsanguemiliussilva2013sound}, and show how Grabmayer and Fokkink's proof can be remoulded to fit the global approach.
	At the end of the latter section, we give general conditions under which such a remoulding of a completeness proof that takes a local approach to a global one is possible. 

	A method is presented at the end of \autoref{sec:a global approach} for turning a distinguished class \(\mathcal C\) from the local approach into a class \(\mathcal C_{loc}\) suitable for a global approach.
	Interestingly, the class \(\mathcal C_{loc}\) of locally \(\mathcal C\) coalgebras is closed under arbitrary coproducts, subcoalgebras, and homomorphic images when \(\mathcal C\) is closed under subcoalgebras and homomorphic images.
	In the case of the prechart functor \(P\), and with \(\mathcal C\) the class of finite well-layered precharts, these structural qualities imply that \(\mathcal C_{loc}\) is a \emph{covariety}, meaning that it can be presented by a predicate on a \emph{cofree} \(P\)-coalgebra in some number \(\kappa\) of colours (that is, a \emph{coequation} in \(\kappa\) colours)~\cite{rutten2000universal}.
	The final \(P\)-coalgebra is a cofree coalgebra in one colour, but a covariety presented by a coequation in one colour is closed under bisimilarity~\cite{gummschroeder2001covarieties}, which we know from \autoref{fig:bad rerouting} is not the case for \(\mathcal C_{loc}\).
	We suspect that the number of colours needed to present the covariety of locally well-layered precharts is infinite, due to the infinitary nature of the layeredness condition, but more work needs to be done to be sure.

	The use value of covarieties in the pursuit of completeness theorems is generally not well-understood.
	From Theorem \ref{thm:local to global}, we expect there to be a deeper connection, but this is something that can only be uncovered by considering more examples.
	For instance, a covariety appears as the distinguished class of automata in the completeness proof in \cite{schmidkappekozensilva2021gkat}, the presenting coequation being the image of the expression language under the final coalgebra homomorphism.
	The situation in \emph{loc. cit.} was similar to Grabmayer and Fokkink's, in that it was a completeness proof which lacked the use of a full Kleene theorem, and so could be an example of the phenomenon we are alluding to.
	We think other examples could be found by giving different operational interpretations of star expression languages considered in the literature, including as~\cite{kozensmith1996kat,kappebrunetsilvazanasi2017concurrent,jipsen2014concurrent,smolkafosterhsukappekozensilva2019gkat,schmidkappekozensilva2021gkat}, as well as their fixed-point versions.
	Furthermore, we generally suspect that when going from semantics to expressions via solutions that depend only on generated subcoalgebras, a coequation should specify the distinguished class of models, thus enabling either of the two approaches to completeness discussed in this paper to go through. 
	

%% file: appendix.tex
\section{Proofs from \autoref{sec:coalgebras and 1-free star expressions}}

	\restatefunctionalbisimulations*
	
	\begin{proof}
		Let \(h : X \to Y\) be a coalgebra homomorphism.
		To see that (i) holds, observe that
		\[
			o_Y \circ h = \pi_1 \circ \langle o_Y, \partial_Y \rangle \circ h
			= \pi_1 \circ (\id_{2^A}\times\Pfin(h)) \circ \langle o_X, \partial_X\rangle
			= o_X,
		\]
		so \(x \Rightarrow a\) if and only if \(h(x) \Rightarrow a\) for any \(x \in X\), \(a \in A\).
		For (ii), observe that
		\[
		\partial_Y \circ h = \pi_2 \circ \langle o_Y, \partial_Y \rangle \circ h
		= \pi_2 \circ (\id_{2^A}\times\Pfin(h)) \circ \langle o_X, \partial_X\rangle
		= \Pfin(h) \circ \partial_X.
		\]
		Now, if \(x \trans{a} x'\), then \(h(x') \in h(\partial_X(x)(a)) = \partial_Y(h(x))(a)\), so \(h(x) \trans{a} h(x')\).
		Conversely, if \(h(x) \trans{a} y'\), then \(y' \in \partial_Y(h(x))(a) = h(\partial_X(x)(a))\), so there is an \(x' \in X\) such that \(x \trans{a} x'\) and \(h(x') = y'\).
		
		Conversely, assume that \(h\) satisfies (i) and (ii). 
		Property (i) says that \(o_Y \circ h = o_X\), and (ii) says that \(\partial_Y \circ h = \Pfin(h) \circ \partial_X\).
		It immediately follows that 
		\[
			(\id_{2^A} \times \Pfin(h)) \circ \langle o_X, \partial_X\rangle 
			= \langle o_X, \Pfin(h) \circ \partial_X\rangle = \langle o_Y \circ h, \partial_Y \circ h\rangle
			= \langle o_Y, \partial_Y\rangle \circ h.
		\]
	\end{proof}
	
	\restatestrongsoundness*
	
	\begin{proof}
		%
		Let \(e,f \in \SExp\) and \(a \in A\).
		Following the characterisation of bisimulations after Lemma \ref{lem:functional bisimulations}, we show by induction on the proof of \(e \bisim f\) that (i) \(e \Rightarrow a\) if and only if \(f \Rightarrow a\), for any \(a \in A\), and (ii) \(e \trans{a} e'\) only if there exists an \(f' \in\SExp\) such that \(f \trans{a} f'\) and \(e'\equiv f'\), and (iii) \(f \trans{a} f'\) only if there exists an \(e' \in\SExp\) such that \(e \trans{a} e'\) and \(e'\equiv f'\). 
		
		If the proof of \(e \equiv f\) is composed of only a single equational axiom, then (i) and (ii) can be shown directly by considering each axiom separately.
		The property (iii) then follows from (ii) by symmetry.
		\begin{itemize}
			\item[(B1)] \(e = e_1 + e_2\) and \(f = e_2 + e_1\) for some \(e_1,e_2 \in \SExp\). 
			For (i), we have \(e \Rightarrow a\) if and only if \(e_1 \Rightarrow a\) or \(e_2 \Rightarrow a\) if and only if \(f \Rightarrow a\).
			For (ii), \(e \trans{a} e'\) if and only if \(e_1 \trans{a} e'\) or \(e_2 \trans{a} e'\) if and only if \(f \trans{a} e'\).
			\item[(B2)] \(e = e_1 + (e_2 + e_3)\) and \(f = (e_1 + e_2) + e_3\) for some \(e_1,e_2 \in \SExp\).
			For (i), we have \(e \Rightarrow a\) if and only if \(e_i \Rightarrow a\) for some \(i \in \{1,2,3\}\) if and only if \(f \Rightarrow a\).
			For (ii), \(e \trans{a} e'\) if and only if \(e_i \trans{a} e'\) for some \(i \in \{1,2,3\}\) if and only if \(f \trans{a} e'\).
			\item[(B3)] \(e = e_1 + e_1\) and \(f = e_1\) for some \(e_1 \in \SExp\).
			For (i), we have \(e \Rightarrow a\) if and only if \(e_1 \Rightarrow a\) if and only if \(f \Rightarrow a\).
			For (ii), \(e \trans{a} e'\) if and only if \(e_1 \trans{a} e'\) if and only if \(f \trans{a} e'\).
			\item[(B4)] \(e = (e_1 + e_2)e_3\) and \(f = e_1e_3 + e_2 e_3\) for some \(e_1,e_2,e_3 \in \SExp\).
			Neither \(e\) nor \(f\) terminate after any action, so (i) holds vacuously.
			For (ii), assume \(e \trans{a} e'\).
			If \(e' = e_i'e_3\) and \(e_i \trans{a} e_i'\) for some \(i \in \{1,2\}\), then \(e_ie_3 \trans{a} e_i'e_3\) for the same \(i\), and therefore \(f \trans{a} e_i'e_3\).
			If \(e_i \Rightarrow a\) and \(e_3 \trans{a} e'\) for some \(i \in \{1,2\}\), then \(e_ie_3 \trans{a} e'\) for the same \(i\), and therefore \(f \trans{a} e'\).
			Conversely, if \(f \trans{a} e'\), then \(e_ie_3 \trans{a} e'\) for some \(i \in \{1,2\}\).
			Thus, for the same \(i\), either \(e_i \trans{a} e_i'\) and \(e' = e_i'e_3\) or \(e_i \Rightarrow a\) and \(e_3 \trans{a} e'\).
			In either case, \(e \trans{a} e'\).
			\item[(B5)] \(e = (e_1e_2)e_3\) and \(f = e_1(e_2e_3)\).
			For (i), same as (B4).
			For (ii), \(e \trans{a} e'\) if and only if either \(e_1 \trans{a} e_1'\) and \(e' = (e_1'e_2)e_3\) or \(e_1 \Rightarrow a\) and \(e_2e_3 \trans{a} e'\).
			In the first case, \(e_1 \trans{a} e_1'\) and \(f \trans{a} e_1'(e_2e_3) \equiv (e_1'e_2)e_3\).
			In the second, \(f \trans{a} e'\) as well.
			The converse holds symmetrically.
			\item[(B6)] \(e = e_1 + 0\) and \(f = e_1\) for some \(e_1 \in \SExp\).
			Since \(0\) does not have outgoing transitions, this case is the same as (B3).
			\item[(B7)] \(e = 0e_1\) and \(f = 0\) for some \(e_1 \in \SExp\).
			Since \(0\) has no outgoing transitions, this case is vacuous.
			\item[(BKS1)] \(e = e_1*e_2\) and \(f = e_1(e_1 * e_2) + e_2\) for some \(e_1,e_2 \in \SExp\).
			For (i), since \(e_1(e_1 * e_2)\) does not terminate after any action, \(f \Rightarrow a\) if and only if \(e_2 \Rightarrow a\) if and only if \(e \Rightarrow a\).
			For (ii), assume \(e \trans{a} e'\).
			If \(e' = e_1*e_2\) and \(e_1 \Rightarrow a\), then \(e_1(e_1*e_2) \trans{a} e_1*e_2\).
			If \(e' = e_1'(e_1*e_2)\) and \(e_1 \trans{a} e_1'\), then \(e_1(e_1*e_2) \trans{a} e'\) and therefore \(f \trans{a} e'\).
			If \(e_2 \trans{a} e'\), then \(f \trans{a} e'\) also.
			Conversely, \(f \trans{a} f'\) if and only if either \(e_1 \trans{a} e_1'\) and \(f' = e_1'(e_1*e_2)\) or \(e_2 \trans{a} f'\) or \(e_1 \Rightarrow a\).
			In any case, \(e \trans{a} f'\).
			\item[(BKS2)] \(e = (e_1*e_2)e_3\) and \(f = e_1*(e_2e_3)\) for some \(e_1,e_2,e_3 \in \SExp\).
			For (i), follow (B4).
			For (ii), assume \(e \trans{a} e'\).
			If \(e_1 \Rightarrow a\) and \(e' = e\), then \(f \trans{a} f\) as well.
			If \(e' = (e_1'(e_1*e_2))e_3\) and \(e_1 \trans{a} e_1'\), then \(f \trans{a} e_1'(e_1*(e_2e_3)) \equiv e'\).
			If \(e_2 \trans{a} e_2'\) and \(e' = e_2'e_3\), then \(f \trans{a} e'\).
			The converse is similar.
		\end{itemize}
	
		The inductive step is broken into four cases, depending on if the last step in the proof of \(e \equiv f\) is reflexivity, symmetry, transitivity, or (RSP).
		The reflexivity and symmetry cases are trivial, by the symmetry between (ii) and (iii).
		For the transitivity rule, let \(e \equiv g\) and \(g \equiv f\) satisfy (i)-(iii).
		Clearly, \(e \Rightarrow a\) if and only if \(f \Rightarrow a\).
		If \(e \trans{a} e'\), then \(g \trans{a} g'\) for some \(g'\) such that \(e' \equiv g'\), and consequently \(f \trans{a} f'\) such that \(g' \equiv f'\).
		Hence, \(e' \equiv f'\), and therefore (ii) holds for \(e \equiv f\).
		Symmetrically, we obtain (iii).
		
		For the (RSP) case, suppose (i)-(iii) hold for \(g \equiv e g + f\).
		Then \(g \Rightarrow a\) if and only if \(f \Rightarrow a\) if and only if \(e * f \Rightarrow a\).
		Furthermore, if \(g \trans{a} g'\), then there is an \(e'\) such that \(eg + f \trans{a} e'\) and \(g' \equiv e'\).
		Whence, either (a) \(e \Rightarrow a\) and \(e' = g\), (b) \(e \trans{a} e''\) and \(e' = e''g\), or (c) \(f \trans{a} f'\). 
		\begin{itemize}
			\item[(a)] Since \(g \equiv e * f\) by (RSP), \(e * f \trans{a} e * f \equiv g\).
			\item[(b)] Since \(g \equiv e * f\) by (RSP), \(g' \equiv e' = e''g \equiv e''(e*f)\).
			Also, \(e*f \trans{a} e''(e*f)\).
			\item[(c)] Simply observe that \(e*f \trans{a} f'\) as well.
		\end{itemize}
		The converse is similar.
		This concludes the proof.
	\end{proof}
	
	\restatefundamentaltheorem*
	
	\begin{proof}
		See Lemma A.2 in \cite{grabmayerfokkink2020extended}.
		We will prove this by induction on the construction of \(e\).
		If \(e = 0\) or \(e = b \in A\), then
		\[
		0 \stackrel{\text{\footnotesize{(B6)}}}{\equiv} 0 + 0 \equiv \sum_{0 \Rightarrow a} a + \sum_{0 \trans{a} f} a~f
		\qquad\qquad
		b \stackrel{\text{\footnotesize{(B6)}}}{\equiv} b + 0 \equiv \sum_{b \Rightarrow a} a + \sum_{b \trans{a} f} a~ f
		\]
		by definition.
		Now, assuming the result for \(e_1,e_2\), there are three cases to consider.
		\begin{itemize}
			\item \(e = e_1 + e_2\).
			\begin{align*}
				e_1 + e_2 &\equiv \left( \sum_{e_1 \Rightarrow a} a + \sum_{e_1 \trans{a} e_1'} a~e_1' \right) + \left( \sum_{e_2 \Rightarrow a} a + \sum_{e_2 \trans{a} e_2'} a~e_2'\right) \\
				&\equiv  \sum_{e_1 + e_2 \Rightarrow a} a + \sum_{e_1 + e_2 \trans{a} e'} a~e_1' \tag{B1,B2}
			\end{align*}
			\item \(e = e_1e_2\).
			\begin{align*}
				e_1e_2 &\equiv \left(\sum_{e_1 \Rightarrow a} a + \sum_{e_1 \trans{a} e_1'} a~e_1'\right)e_2\\
				&\equiv \sum_{e_1 \Rightarrow a} a~e_2 + \sum_{e_1 \trans{a} e_1'} a~e_1'e_2 \tag{B2-B4}
			\end{align*}
			\item \(e = e_1 *e_2\).
			\begin{align*}
				e_1 * e_2 &\equiv e_1(e_1 * e_2) + e_2 \tag{BSK1}\\
				&\equiv \left(\sum_{e_1 \Rightarrow a} a + \sum_{e_1 \trans{a} e_1'} a~e_1'\right)(e_1 * e_2) + \left( \sum_{e_2 \Rightarrow a} a + \sum_{e_2 \trans{a}  e_2'} a~e_2'\right)\\
				&\equiv \sum_{e_1 \Rightarrow a} a~(e_1 * e_2) + \sum_{e_1 \trans{a} e_1'} a~e_1'(e_1 * e_2) + \sum_{e_2 \Rightarrow a} a + \sum_{e_2 \trans{a} e_2'} a~e_2' \tag{B2-B4}\\
				&\equiv \sum_{e_2 \Rightarrow a} a + \sum_{e_1 \Rightarrow a} a~(e_1 * e_2) + \sum_{e_1 \trans{a} e_1'} a~e_1'(e_1 * e_2) +  \sum_{e_2 \trans{a} e_2'} a~e_2'\tag{B1}\\
			\end{align*}
		\end{itemize} 
	\end{proof}
	
	\restatesolutionlemma*
	
	\begin{proof}
		Let \(s : X \to \SExp\), and \(R = \{(x, [s(x)]_{\equiv}) \mid x \in X\}\).
		Then \([-]_{\equiv} \circ s\) is a homomorphism if and only if \(R\) is a bisimulation.
	
		If \(R\) is a bisimulation, then \(x \Rightarrow a\) if and only if \(s(x) \Rightarrow a\), and \(s(x) \trans{a} e'\) if and only if there is an \(x' \in X\) such that \(x \trans{a} x'\) and \(s(x') \equiv e'\), because \([-]_{\equiv}\) is a coalgebra homomorphism.
		Thus, by (B3), \(\sum \partial(s(x))(a) \equiv \sum s(\partial(x)(a))\), so that\[
			s(x) \equiv \sum_{s(x) \Rightarrow a} a + \sum_{s(x) \trans{a} e'} a~e' \equiv \sum_{x \Rightarrow a} a + \sum_{x \trans{a} x'} a~s(x').
		\]
		It follows that \(s\) is a solution to \(X\).
	
		Conversely, assume \(s\) is a solution to \(X\), \[
			s(x) \equiv \sum_{x \Rightarrow a} a + \sum_{x \trans{a} x'} a~s(x').
		\]
		Then \(x \Rightarrow a\) if and only if \(s(x) \Rightarrow a\), and \(s(x) \trans{a} e'\) if and only if there is an \(x' \in X\) such that \(x \trans{a} x'\) and \(s(x') \equiv e'\), because \(\equiv\) is a bisimulation.
		Whence, \(R\) is a bisimulation because \([-]_{\equiv}\) is a coalgebra homomorphism.
	\end{proof}

\section{Proofs from \autoref{sec:coalgebras and completeness}}

	\restatecompletenessstrategy*
	
	\begin{proof}
		Let \(e, f \in E\).
		Since \(\equiv\) is a bisimulation equivalence, \(e \equiv f\) implies \(e \bisim f\) for any \(e,f \in E\) by definition. 
		Therefore, it suffices to show the converse.
	
		Suppose \(e \bisim f\), and let \(X = \langle e \rangle\) and \(Y = \langle f \rangle\).
		As \(\mathcal C\) is collapsable, there is a \(Z \in \mathcal C\) and a pair of homomorphisms \(p : X \to Z, q : Y \to Z\) such that \(p(e) = q(f)\).
		Let \(s : Z \to E/{\equiv}\) be a homomorphism, and consider the diagram
		\[\begin{tikzcd}[ampersand replacement=\&]
			X \ar[d, hook] \ar[r, "{p}"] \& Z \ar[d, "s"] \& Y \ar[l, "q"'] \ar[d, hook]\\
			E \ar[r, "{[-]_{\equiv}}"] \& E/{\equiv} \& E \ar[l, "{[-]_{\equiv}}"']
		\end{tikzcd}\]
		Since \(X\) and \(Y\) admit a unique homomorphism into \(E/{\equiv}\), this diagram commutes.
		In particular, we have \([e]_{\equiv} = s(p(e)) = s(q(f)) = [f]_{\equiv}\), meaning \(e \equiv f\).
	\end{proof}
	
	\begin{remark}
		The assumption that \(\langle e\rangle \in \mathcal C\) for all \(e \in E\) is also not completely necessary.
		We could instead assume that \(E/{\equiv} = \bigcup \{s(X) \mid \text{\(s\) solves \(X \in \mathcal C\)}\}\).
		The proof of Theorem \ref{thm:completeness strategy} would then take \(X\) and \(Y\) to be any coalgebras in \(\mathcal C\) such that \(s_X(x) = [e]_{\equiv}\) and \(s_Y(y) = [f]_\equiv\) for some \(x \in X\) and \(y \in Y\) and solutions \(s_X\) and \(s_Y\) to \(X\) and \(Y\) respectively.
	\end{remark}

\section{Proofs from \autoref{sec:layered loop existence and elimination}}

	\begin{proposition}
		A locally finite prechart \(X\) is well-layered if and only if it has an \acro{LLEE}-witness \cite{grabmayerfokkink2020complete}.
	\end{proposition}

	\begin{proof}
		To obtain an \acro{LLEE}-witness from a layering witness \(X^\bullet\), let \(\hat X\) be the numerical entry/body labelling of \(X\) (as in \cite{grabmayerfokkink2020complete}) obtained as follows: label each \(x \trans{a} y\) in \(X\) such that \(x\eo y\) in \(X^\bullet\) with a weighted transition \(x \trans{a}_{[|x|_{en}]} y\), where\[
		|x|_{en} = \max\{m \in \N \mid \text{\((\exists x_1,\dots,x_m)~x \diredge x_1 \diredge \cdots \diredge x_m\) s.t. \(x \neq x_i \neq x_j\) for \(i \neq j\)}\},
		\] 
		and each \(x \trans{a} y\) such that \(x \bo y\) with \(x \trans{a}_{[0]} y\). 
		This is a well-defined translation because we have assumed that \(\langle x \rangle\) is finite and \((\langle x \rangle, \diredge)\) is acyclic.
		We proceed to check each of the \acro{LLEE}-witness conditions.
		\begin{enumerate}
			\item[(W1)] Since \(X\) is locally finite, every infinite path in \(X\) contains a cycle.
			It follows that no body loop can exist in \(\hat X\), because \(X^\bullet\) is fully specified and therefore has no body cycles.
			\item[(W2)] Let \(v \in X\) and \(n \in \N\), and recall that \(X_{\hat X}(v,n)\) is the union of all paths of the form \[
			v \trans{a}_{[n]} x_1 \trans{b_1}_{[0]} x_2 \trans{\ }_{[0]} \cdots \trans{ }_{[0]} x_m \trans{b_m}_{[0]} y
			\]
			such that \(v \nin \{x_1,\dots, x_m\}\).
			For \(n \neq |v|_{en}\), \(X_{\hat X}(v, n)\) is the single vertex \(v\), so it suffices to show that (a) \(X_{\hat X}(v, |v|_{en})\) is a loop chart when \(|v|_{en}>0\), and (b) every transition in \(\hat X\) whose source in \(X_{\hat X}(v, n)\) is not equal to \(v\) carries a numerical label strictly less than \(|v|_{en}\).
			\begin{itemize}
				\item[(a)] This consists of three parts, corresponding to the loop chart conditions (L1), (L2), and (L3) in \cite{grabmayerfokkink2020complete}.
				The condition (L1) follows from full specification, since if \(v \trans{a}_{[|v|_{en}]} y\), then \(y \to^+ v\).
				This implies that there is a minimal cycle (a cycle \(x_1 \to \cdots \to x_k\) such that \(x_i \neq x_j\) when \(i \neq j\)) containing \(v \eo y\), in which every other transition must be a body transition by layeredness.
				Condition (L2) follows from the lack of body cycles in \(X^\bullet\): Since an infinite path in \(X\) must contain a minimal cycle, and every cycle in \(X_{\hat X}(v,|v|_{en})\) contains \(v\), every infinite path from \(v\) in \(X_{\hat X}(v,|v|_{en})\) eventually passes through \(v\). 
				Condition (L3) follows from \(X^\bullet\) being goto-free.
				\item[(b)] This is a direct consequence of \(X\) being layered, since \((X^\bullet, \diredge)\) and \((\hat X, \diredge)\) coincide.
			\end{itemize}
		\end{enumerate} 
		
		To obtain a layering witness \(X^\bullet\) from an \acro{LLEE}-witness \(\hat X\), replace each \(x \trans{a}_{[n]} y\) in \(\hat X\) by \(x \eo y\) if \(n > 0\) and \(y \to^+ x\), or by \(x \bo y\) otherwise.
		We now proceed to check the conditions in Definition~\ref{def:layering witness}.
		\begin{enumerate}
			\item[(i)] This entry/body labelling is locally finite because \(X\) is.
			\item[(ii)] That \(X^\bullet\) is flat can be shown in two steps.
			First, observe that every transition \(x \trans{a}_{[n]} y\) such that \(x \eo y\) in \(X^\bullet\) appears in a cycle, and therefore also a minimal cycle \(L\) in \(\hat X\).
			Since \(L\) is minimal and \(\hat X\) is layered, there is exactly one transition \(x' \trans{b}_{[m]} y'\) in \(L\) with \(m > 0\).
			Every non-\((x \to y)\) transition in \(L\) is a \(x \bo y\) transition in \(X^\bullet\), so \(x \trans{a}_{[n]} y\) must carry a numerical label \(n > 0\), and therefore must be the unique entry transition in \(L\).
			If it is also the case that \(x \trans{b}_{[m]} y\), and we replace \(x \trans{a}_{[n]} y\) by \(x \trans{b}_{[m]} y\) in \(L\) to obtain \(L'\), then either \(m > 0\) or \(L'\) is a body cycle.
			We have assumed \(\hat X\) is free of body cycles, so \(m > 0\) if and only if \(n > 0\).
			This establishes the flatness of \(X^\bullet\).
			\item[(iii)] Since \(x \eo y\) implies \(y \to^+ x\) by construction, it suffices to see that \(\neg (x \bo^+ x)\) for all \(x \in X\).
			So, assume that \(x \to^+ x\) in \(X\).
			Then there is a minimal cycle \(L\) in \(\hat X\) containing \(x\).
			Then \(L\) must contain a transition of the form \(x' \trans{a}_{[n]} y'\) with \(n > 0\) by (W2)(b). 
			Since \(y' \to^+ x'\), \(x' \eo y'\) in \(X^\bullet\).
			This means that \(L\) does not correspond to a body loop in \(X^\bullet\), so \(\neg(x \bo^+ x)\).
			\item[(iv)] \(X^\bullet\) is layered because \((X^\bullet, \diredge)\) is a subgraph of \((\hat X, \diredge)\) and \((\hat X, \diredge)\) is acyclic.
			\item[(v)] \(X^\bullet\) is goto-free by (W2)(a)(L3) from \cite{grabmayerfokkink2020complete}.
		\end{enumerate}
	\end{proof}

	\begin{lemma}
		If \(X\) is a well-layered prechart and \(U \subseteq X\) is a subcoalgebra of \(X\), then \(U\) is well-layered.
	\end{lemma}
	
	\begin{proof}
		Let \(X^\bullet\) be a layering witness for \(X\).
		Since \(U\) is a subcoalgebra of \(X\), and \(X^\bullet\) and \(X\) share an underlying \textsf{tswo}, restricting \(X^\bullet\) to the states in \(U\) gives a subcoalgebra \(U^\bullet\) of \(X^\bullet\).
		It suffices to see that \(U^\bullet\) is a layering witness for \(U\).
		
		It is easily seen that \(U^\bullet\) is locally finite and flat.
		Furthermore, \(U^\bullet\) is closed under the transitions of \(X^\bullet\), so \(U^\bullet\) is easily seen to be fully specified as well.
		To see that \(U^\bullet\) is layered, it suffices to see that \((U, \diredge)\) is a directed subgraph of \((X,\diredge)\).
		To this end, simply observe that \(x \diredge y\) in \(U^\bullet\) if and only if \(x \in U\) and \(x \diredge y\) in \(X^\bullet\), since \(U^\bullet\) is closed under the transitions of \(X^\bullet\).
		Similarly, \(U^\bullet\) is goto-free because \(x \diredge y \Rightarrow\) in \(U^\bullet\) if and only if \(x \in U\) and \(x \diredge y \Rightarrow\) in \(X^\bullet\) because \(U^\bullet\) is closed under the transitions of \(X^\bullet\).
	\end{proof}

	\begin{lemma}\label{lem:one_of_these_forms}
		Every minimal loop in \(\SExp\) is of the form
		\begin{equation*}
			(e*f)g_1\dots g_l \to e_1(e*f)g_1\dots g_l \to \cdots \to e_m(e*f)g_1\dots g_l \to (e*f)g_1\dots g_l \\
		\end{equation*}
		for some \(e,e_i,f,g_i \in \SExp\) and \(l \ge 0\), where the sequential composition operation associates left by default.
	\end{lemma}

	\begin{proof}
		Let \(h \in \SExp\).
		By induction on the construction of \(h\), we begin by showing that if
		\begin{equation}\label{eq:h_0 loop}
			h \to^+ h_0 \to h_1 \to \cdots \to h_m \to h_0
		\end{equation}
		in \(\SExp\), then every \(h_i \in \{(e*f)g_1\dots g_l, e_i(e*f)g_1\dots g_l\}\) for some fixed \(e,e_i,f,g_i \in \SExp\).
		The base case is vacuous, for if \(h = 0\) or \(h = a \in A\), then no such loop in \eqref{eq:h_0 loop} exists.
		For the inductive step, assume that the above holds for some \(j\) and \(k\) in \(\SExp\).
		There are three cases to consider.
		\begin{itemize}
			\item[(1)] In case \(h = j + k\), and given a loop such as \eqref{eq:h_0 loop}, either \(j \to^+ h_0\) or \(k \to^+ h_0\).
			By the induction hypothesis, every \(h_i \in \{e_i(e*f)g_1\dots g_l, e_i(e*f)g_1\dots g_l\}\) for some \(e,e_i,f,g_i \in \SExp\), as desired.

			\item[(2)] In case \(h = jk\), either \(k \to^+ h_i\) for some \(i\), or \(h_i = j_ik\) for all \(i \le m\) and\[j \to^+ j_0 \to j_1 \to \cdots \to j_m \to j_0,\]
			for some \(j_0,\dots,j_m \in \SExp\).
			In the former case, we simply appeal to the inductive hypothesis again.
			In the latter, observe that by the inductive hypothesis \(j_i \in \{(e*f)g_1\cdots g_l, e_i(e*f)g_1\cdots g_l\}\) for each \(i\).
			Whence, \(h_i \in \{(e*f)g_1\dots g_lk, e_i(e*f)g_1\dots g_lk\}\), as desired. 

			\item[(3)] In case \(h = j*k\), either (a) \(k \to^+ h_i\) for some \(i \le m\), or (b) \(j = j_0 \to j_1 \to\cdots\to j_m \Rightarrow\) and \(h_i = j_i(j*k)\), or (c) \(j \to^+ j_0 \to \cdots \to j_m \to j_0\) as in (2).
			Sub-case (a) is settled by the induction hypothesis.
			In (b), every of the expressions is of the desired form (here, \(l = 0\)), so we are done.
			Finally, in (c), by the induction hypothesis \(j_i \in \{(e*f)g_1\dots g_l, e_i(e*f)g_1\dots g_l\}\) for each \(i\).
			Hence, \(h_i = j_i(j*k) \in \{(e*f)g_1\dots g_l(j*k), e_i(e*f)g_1\dots g_l(j*k)\}\) for each \(i\) as desired.
		\end{itemize}

		Suppose now that \(h_0 \to h_1 \to \cdots \to h_m \to h_0\)	is a minimal cycle in \(\SExp\).
		By the previous observation, we may assume without loss of generality this cycle is of the form
		\begin{mathpar}
			e_0(e*f)g_1\dots g_l \to e_1(e*f)g_1\dots g_l \to \cdots \to (e*f)g_1\dots g_l \to e_p(e*f)g_1\dots g_l \to \cdots \\ \cdots\to e_m(e*f)g_1\dots g_l \to e_0 (e*f)g_1\dots g_l
		\end{mathpar}
		with \(e_0 \to e_1 \to \cdots \to e_m \to e_0\).
		The desired result can be seen simply by shifting the indices ahead by \(p\) modulo \(m+1\).
	\end{proof}

	\restatestarexpressionwelllayered*

	\begin{proof}
		First, we check that \(\SExp\) is locally finite.
		To this end, define \(N : \SExp \to \N\) inductively, as
		\begin{mathpar}
			N(0) = N(a) = 1
			\and
			N(e_1 + e_2) = N(e_1*e_2) = N(e_1) + N(e_2)
			\and
			N(e_1e_2) = N(e_1) + N(e_1)N(e_2)
		\end{mathpar}
		It follows by induction on \(e\) that the number of states appearing in \(\langle e \rangle\) is at most \(N(e)\).
		If \(e = 0\) or \(e = a \in A\), then clearly \(|X| \le 1 = N(e)\).
		For the inductive step, let \(X_i = \langle e_i \rangle\) for \(i = 1,2\) and \(X = \langle e \rangle\), and \(|-|\) count the states in a coalgebra.
		There are three cases to consider.
		\begin{itemize}
			\item \(e = e_1 + e_2\).
			\begin{align*}
				|X| &= |\{e' \mid e_1 \to^+ e'\text{ or } e_2 \to^+ e' \text{ or } e' = e_1 + e_2\}| \\
				&\le |X_1| + |X_2| \le N(e_1) + N(e_2) = N(e_1 + e_2) =N(e). 
			\end{align*}
			\
			\item \(e = e_1e_2\).
			\begin{align*}
				|X| &= |\{e' \mid \text{\(e' = e_1'e_2\) and \(e_1 \to^+ e_1'\), or \(e_2 \to^+ e'\) and \(e_1 \to^*\Rightarrow a\), or \(e' = e\)}\}|\\
				&\le |X_1| + |X_1||X_2| \le N(e_1) + N(e_1)N(e_2) = N(e).
			\end{align*}
			\item \(e = e_1*e_2\).
			\begin{align*}
				|X| &= |\{e' \mid \text{\(e' = e_1'(e_1*e_2)\) and \(e_1 \to^+ e_1'\), or \(e_2 \to^+ e'\), or \(e' = e\)}\}|\\
				&\le |X_1| + |X_2| \le N(e_1) + N(e_2) = N(e).
			\end{align*}
		\end{itemize}
		This shows that \(\SExp\) is locally finite, and thus any entry/body labelling of \(\SExp\) must also be locally finite.

		Next, we provide \(\SExp\) with a layering witness.
		Define \(\SExp^\bullet\) via the proof rules below:
		\begin{mathpar}
			\infer{e_i \to f}{e_1 + e_2 \bo f}
			\and
			\infer{e_1 \eo f}{e_1e_2 \eo fe_2}
			\and
			\infer{e_1 \bo f}{e_1e_2 \bo fe_2}
			\and
			\infer{e_1 \Rightarrow}{e_1e_2 \bo e_2}
			\and
			\infer{e_1 \Rightarrow}{e_1*e_2 \eo e_1*e_2}
			\and
			\infer{e_2 \to f}{e_1*e_2 \bo f}
			\and
			\infer{e_1 \to f \and (\exists g)\ e_1 \to^* g \Rightarrow}{e_1*e_2 \eo f(e_1*e_2)}
			\and
			\infer{e_1 \to f \and \neg ((\exists g)\ e_1 \to^* g \Rightarrow)}{e_1*e_2 \bo f(e_1*e_2)}
		\end{mathpar}
		Here, \(e \to f\) if either \(e \eo f\) or \(e \bo f\), and \(e \Rightarrow\) in \(\SExp^\bullet\) if and only if \(e \Rightarrow\) in \(\underline{\SExp}\).
		It is easy to see this labelling is flat.
		By Lemma~\ref{lem:one_of_these_forms}, every minimal loop in \(\SExp\) is of the form
		\begin{equation*}\label{eq:loop form}
			(e*f)g_1\dots g_l \to e_1(e*f)g_1\dots g_l \to \cdots \to e_m(e*f)g_1\dots g_l \to (e*f)g_1\dots g_l ,\\
		\end{equation*}
		where \(e \to e_1 \to \cdots \to e_m \Rightarrow\).
		In \eqref{eq:loop form}, the initial transition is an \(\eo\)-transition, so \(\SExp^\bullet\) satisfies \(\neg(x \bo^+ x)\) for all \(x \in \SExp\).
		Conversely, an easy induction on the proof of \(e \eo f\) reveals that there is an \(l \ge 0\) such that \(e = (e_1*f_1)g_1\dots g_l\) for some \(e_1,f_1,g_i \in \SExp\), and such that either \(f = e_1'e\) and \(e_1' \to^+ \Rightarrow\), or \(f = e\).
		In the first case, since \(e_1' \to^+\Rightarrow\), there is a path \(f \to^+ e\).
		Note that there are no paths \(f \to^+ \Rightarrow\) that do not pass through \(e\), as \(e_1e_2\) does not terminate for any \(e_1,e_2 \in \SExp\), thus \(\SExp^\bullet\) is goto-free.
		In the second case, \(f \to e\) by assumption.
		This establishes that \(\SExp^\bullet\) is fully specified and goto-free.

		Lastly, we check that \(\SExp^\bullet\) is layered.
		For this, we follow \cite{grabmayerfokkink2020complete} and define the \define{loop depth} \(|e \to f|_{ld}\) of a transition in an entry/body labelling by the following rules:
		\begin{mathpar}
			\infer{e \bo f}{|e \to f|_{ld} = 0}
			\and
			\infer{|e_1 \to f|_{ld} = n}{|e_1e_2 \to fe_2|_{ld} = n}
			\and
			\infer{e_1 \to f \and (\exists g)\ e_1 \to^* g \Rightarrow}{|e_1*e_2 \to f(e_1*e_2)|_{ld} = |e_1|_* + 1}
			\and
			\infer{e_1 \Rightarrow}{|e_1*e_2 \to e_1*e_2|_{ld} = |e_1|_* + 1}
		\end{mathpar}
		where \(|e|_*\) is the \define{star height} of the expression \(e\), defined as
		\begin{mathpar}
			|0|_* = |a|_* = 0
			\and
			|e + f|_* = |ef|_* = \max\{|e|_*, |f|_*\}
			\and
			|e*f|_* = 1 + \max\{|e|_*, |f|_*\}
		\end{mathpar}
		It now suffices to see that for any minimal path\[
			e_1 \eo e_1' \bo e_1'' \bo \cdots \eo e_2 \bo e_2' \bo \cdots \bo e_3 \eo \cdots \bo e_n,
		\]
		we find \(|e_{i+1} \to e_{i+1}'|_{lb} < |e_{i} \to e_{i}'|_{lb}\) for each \(i\).
		For a fixed \(i\), \(e_i\) is necessarily of the form \((e * f)g_1\dots g_l\), and \(e_i' = e'(e * f)g_1\dots g_l\) for \(e \to e'\), and \(e \to^* \Rightarrow\).
		Moreover, \(e_{i+1} = e^{(k)}(e*f)g_1\dots g_l\) where \(e \to e' \to \cdots \to e^{(k)}\).
		So, it suffices to see that \(|e^{(k)}|_* \le |e|_*\), since \(|e_{i+1} \to e_{i+1}'|_{ld} \le |e^{(k)}|_* -1\).
		This can be shown by induction on construction of \(e\) as a term in \(\SExp\).
	\end{proof}

	\restateuniquesolutionslemma*
	
	\noindent\emph{Proof.} 
	We closely follow Grabmayer and Fokkink's proof in \cite{grabmayerfokkink2020extended}.
	Let \(X^\bullet\) be a layering witness for \(X\), and \(s_X\) be the map defined in \eqref{eq:canonical solution}.
	We need to do two things: First, show that \(s_X\) is a solution to \(X\), and second, show that every other solution \(s : X \to \SExp\) satisfies \(s(x) \equiv s_X(x)\) for all \(X\).
	Note that this implicitly proves that \(s_X\) does not depend on \(X^\bullet\), up to \({\equiv}\).
	We need the following lemma.
	
	\begin{lemma}[Lemma 5.4 of \cite{grabmayerfokkink2020extended}]\label{lem:first step to uniq}
		For any \(x \diredge y\) in \(X^\bullet\), \(s_X(y) \equiv t_X(y, x)~s_X(x)\).
	\end{lemma}

	\begin{proof}
		The proof proceeds by induction on \(|y|_b = \max\{m \mid (\exists y_i)~y \bo y_1 \bo \cdots \bo y_m\}\).
		If \(|y|_b = 0\), then 
		\begin{align*}
			s_X(y)
			\equiv \Big(\sum_{y \trans{a}_{\mathsf e} x} a + \sum_{\substack{y\trans{a}_{\mathsf e} v\\y \neq v}} a~t_X(v,y)\Big)*0
			\equiv \Big(\sum_{y \trans{a}_{\mathsf e} x} a + \sum_{\substack{y\trans{a}_{\mathsf e} v\\y \neq v}} a~t_X(v,y)\Big)*0~s_X(x)
			\equiv t_X(y,x) ~s_X(x)
		\end{align*}
		Now, since \(X^\bullet\) is goto-free, \(\neg(y \Rightarrow)\).
		It is not hard to see that if \(y \bo v\), then \(|v|_b < |y|_b\).
		From the induction hypothesis, we obtain
		\begin{align*}
			s_X(y)
			&\equiv \Big(\sum_{y \trans{a}_{\mathsf e} x} a + \sum_{\substack{y\trans{a}_{\mathsf e} v\\y \neq v}} a~t_X(v,y)\Big)
			*
			\Big(0 + \sum_{y\trans{a}_{\mathsf b} v} a~s_X(v)\Big)\\
			&\equiv \Big(\sum_{y \trans{a}_{\mathsf e} x} a + \sum_{\substack{y\trans{a}_{\mathsf e} v\\y \neq v}} a~t_X(v,y)\Big)
			*
			\Big(\sum_{y\trans{a}_{\mathsf b} x} a~s_X(x) + \sum_{\substack{y\trans{a}_{\mathsf b} v\\x \neq v}} a~s_X(v)\Big)\\
			&\equiv \Big(\sum_{y \trans{a}_{\mathsf e} x} a + \sum_{\substack{y\trans{a}_{\mathsf e} v\\y \neq v}} a~t_X(v,y)\Big)
			*
			\Big(\sum_{y\trans{a}_{\mathsf b} x} a~s_X(x) + \sum_{\substack{y\trans{a}_{\mathsf b} v\\x \neq v}} a~t_X(v,x)~s_X(x)\Big)\\	
			&\equiv \Big(\sum_{y \trans{a}_{\mathsf e} x} a + \sum_{\substack{y\trans{a}_{\mathsf e} v\\y \neq v}} a~t_X(v,y)\Big)
			*
			\Big(\sum_{y\trans{a}_{\mathsf b} x} a + \sum_{\substack{y\trans{a}_{\mathsf b} v\\x \neq v}} a~t_X(v,x)\Big)s_X(x)\\	
			&\equiv t_X(y,x)~s_X(x).
		\end{align*}
	\end{proof}

	This allows us to prove that \(s_X\) is a solution to \(X\) as follows: For any \(x \in X\),
	\begin{align*}
		s_X(x) &\equiv \Big(\sum_{x \trans{a}_{\mathsf e} x} a + \sum_{\substack{x\trans{a}_{\mathsf e} y\\x \neq y}} a~t_X(y,x)\Big)
		*
		\Big(\sum_{x \Rightarrow a} a + \sum_{x\trans{a}_{\mathsf b} y} a~s_X(y)\Big)	\tag{def.}\\
		&\equiv \Big(\sum_{x \trans{a}_{\mathsf e} x} a + \sum_{\substack{x\trans{a}_{\mathsf e} y\\x \neq y}} a~t_X(y,x)\Big) s_X(x)
		+ \sum_{x \Rightarrow a} a + \sum_{x\trans{a}_{\mathsf b} y} a~s_X(y) \tag{BSK1}\\
		&\equiv \sum_{x \Rightarrow a} a + \sum_{x \trans{a}_{\mathsf e} x} a~s_X(x) + \sum_{\substack{x\trans{a}_{\mathsf e} y\\x \neq y}} a~t_X(y,x)~s_X(x)
		+  \sum_{x\trans{a}_{\mathsf b} y} a~s_X(y)\tag{B1,4}\\
		&\equiv \sum_{x \Rightarrow a} a + \sum_{x \trans{a}_{\mathsf e} x} a~s_X(x) + \sum_{\substack{x\trans{a}_{\mathsf e} y\\x \neq y}} a~s_X(y)
		+  \sum_{x\trans{a}_{\mathsf b} y} a~s_X(y)\tag{Lemma~\ref{lem:first step to uniq}}\\
		&\equiv \sum_{x \Rightarrow a} a + \sum_{x \trans{a} y} a~s_X(y).
	\end{align*}
	We now turn to uniqueness, for which we need the next lemma.
	
	\begin{lemma}[Lemma 5.7 of \cite{grabmayerfokkink2020extended}]
		For any \(x \diredge y\) in \(X^\bullet\), and any solution \(s : X\to \SExp\) to \(X\), \(s(y) \equiv t_X(y,x)~s(x)\).
	\end{lemma}

	\begin{proof}
		Here, we proceed by induction on \((|x|_{en},|y|_b)\) wrt. the lexicographical ordering on \(\N\), where \(|x|_{en} = \max\{m \mid (\exists y_i)~x \diredge y_1 \diredge \cdots \diredge y_m\}\).
		Now, if \(y \eo v\), then \(|v|_{en} < |y|_{en}\) and consequently \((|y|_{en}, |v|_b) < (|x|_{en}, |y|_{b})\).
		Since \(x \diredge y\), \(\neg(y \Rightarrow)\), and applying the induction hypothesis yields
		\begin{align*}
			s(y) &\equiv 0 + \sum_{y \trans{a} v} a~s(v)\\
			&\equiv  \sum_{y \trans{a}_{\mathsf{e}} y} a~s(y) 
			+ \sum_{\substack{y \trans{a}_{\mathsf{e}} v \\ y \neq v}} a~s(v) 
			+ \sum_{y \trans{a}_{\mathsf{b}} x} a~s(x) 
			+ \sum_{\substack{y \trans{a}_{\mathsf{b}} v \\ x \neq v}} a~s(v) \\
			&\equiv  \sum_{y \trans{a}_{\mathsf{e}} y} a~s(y) 
			+ \sum_{\substack{y \trans{a}_{\mathsf{e}} v \\ y \neq v}} a~t_X(v,y)~s(y) 
			+ \sum_{y \trans{a}_{\mathsf{b}} x} a~s(x) 
			+ \sum_{\substack{y \trans{a}_{\mathsf{b}} v \\ x \neq v}} a~t_X(v,x)~s(x) \\
			&\equiv \Big(\sum_{y \trans{a}_{\mathsf{e}} y} a + \sum_{\substack{y \trans{a}_{\mathsf{e}} v \\ y \neq v}} a~t_X(v,y)\Big)s(y)
			+ \Big(\sum_{y \trans{a}_{\mathsf{b}} x} a + \sum_{\substack{y \trans{a}_{\mathsf{b}} v \\ x \neq v}} a~t_X(v,x)\Big) s(x) \\	
			&\equiv	\Big(\sum_{y \trans{a}_{\mathsf{e}} y} a + \sum_{\substack{y \trans{a}_{\mathsf{e}} v \\ y \neq v}} a~t_X(v,y)\Big)
			*
			\Big(\sum_{y \trans{a}_{\mathsf{b}} x} a + \sum_{\substack{y \trans{a}_{\mathsf{b}} v \\ x \neq v}} a~t_X(v,x)\Big) s(x)\\
			&\equiv t_X(y,x)~s(x)	
		\end{align*}
	\end{proof}

	Uniqueness of solutions can now be proven as follows: Let \(s : X \to \SExp\) be any solution to \(X\), and let \(X^\bullet\) be a layering witness for \(X\).
	We show that \(s(x) \equiv s_X(x)\) by induction on \(|x|_b\) as follows.
	\begin{align*}
		s(x) &\equiv \sum_{x\Rightarrow a} a + \sum_{x \trans{a} y} a~s(y)\\
		&\equiv \sum_{x\Rightarrow a} a
		+\sum_{x \trans{a}_{\mathsf{e}} x} a~s(x) 
		+ \sum_{\substack{x \trans{a}_{\mathsf{e}} y \\ x \neq y}} a~s(y) 
		+ \sum_{x \trans{a}_{\mathsf{b}} y} a~s(y) \\
		&\equiv \sum_{x\Rightarrow a} a
		+\sum_{x \trans{a}_{\mathsf{e}} x} a~s(x) 
		+ \sum_{\substack{x \trans{a}_{\mathsf{e}} y \\ x \neq y}} a~t_X(y,x)~s(x) 
		+ \sum_{x \trans{a}_{\mathsf{b}} y} a~s_X(y) \\
		&\equiv \sum_{x \trans{a}_{\mathsf{e}} x} a~s(x) 
		+ \sum_{\substack{x \trans{a}_{\mathsf{e}} y \\ x \neq y}} a~t_X(y,x)~s(x) 
		+ \sum_{x\Rightarrow a} a
		+ \sum_{x \trans{a}_{\mathsf{b}} y} a~s_X(y) \\
		&\equiv \Big(\sum_{x \trans{a}_{\mathsf{e}} x} a 
		+ \sum_{\substack{x \trans{a}_{\mathsf{e}} y \\ x \neq y}} a~t_X(y,x)\Big)s(x) 
		+ \sum_{x\Rightarrow a} a
		+ \sum_{x \trans{a}_{\mathsf{b}} y} a~s_X(y) \\
		&\equiv \Big(\sum_{x \trans{a}_{\mathsf{e}} x} a 
		+ \sum_{\substack{x \trans{a}_{\mathsf{e}} y \\ x \neq y}} a~t_X(y,x)\Big)* 
		\Big(\sum_{x\Rightarrow a} a
		+ \sum_{x \trans{a}_{\mathsf{b}} y} a~s_X(y)\Big) \\
		&\equiv s_X(x). \tag*{\qed}
	\end{align*}

	\subsection{Reroutings and closure under homomorphic images}

	\restatearbitrarybisimulations*

	\newcommand \scc {\operatorname{scc}}

	\begin{proof}
		We closely follow Grabmayer and Fokkink's proof of Proposition 6.4 of \cite{grabmayerfokkink2020complete}, making the necessary comments as we go along.

		Let \(R\) be a nontrivial bisimulation equivalence on \(X\), and \((u_1,u_2) \in R\) a pair of distinct states.
		For any \(x \in X\), let \(\scc(x) = \{x' \in X \mid x \to^* x'\ \text{and}\ x' \to^* x\}\), called the \emph{strongly connected component} of \(x\).
		We split the proof into two cases, \(\scc(u_1) = \scc(u_2)\), and \(\scc(u_1) \neq \scc(u_2)\).

		Assume \(\scc(u_1) \neq \scc(u_2)\), and without loss of generality that \(\neg (u_2 \to^* u_1)\).
		We argue that \((u_1,u_2) \to^* (w_1,w_2)\) in \(R\) such that \((w_1,w_2)\) satisfies (C1), by induction on \[
			|u_1|_{lb} = \max\left\{n \mathbin{\Big|} (\exists x_i)~\begin{array}{c}
			u_1 \bo x_1 \bo \cdots \bo x_n,\ u_1 \neq x_i \neq x_j\ \text{for}\ i\neq j,\\ \text{and}\
			\scc(u_1) = \scc(x_i)\ \text{for all}\ i
			\end{array}\right\}.
		\]  
		This number is well-defined because \(X\) is locally finite and fully specified.

		For the base case \(|u_1|_{lb} = 0\), assume that \(v \diredge u_1\) for some \(v \in X\).
		If \(u_2 \to^* \bullet \Rightarrow\), then \((u_1,u_2) \to^* \bullet \Rightarrow\), and therefore \(u_1 \to^* \bullet\Rightarrow\) as well.
		Since every body transition out of \(u_1\) leaves \(\scc(u_1)\), \(\neg(u_1 \loopright v)\) for any \(v\), meaning that \(u_1 \bo^* \bullet\Rightarrow\).
		This contradicts the assumption that \(X^\bullet\) is goto-free, so it must have been that \(\neg(u_2 \to^* \bullet \Rightarrow)\).
		Hence, with \(w_1 = u_1\) and \(w_2 = u_2\), the pair \((w_2,w_2)\) satisfies (C1).

		For the induction case, let \(|u_1|_{lb} > 0\), so that \(u_1 \bo u_1'\) for some \(u_1' \in \scc(u_1)\). 
		Since \(R\) is a bisimulation, \((u_1, u_2) \to (u_1', u_2')\) in \(R\), for some \(u_2' \in X\). 
		Since \(|u_1'|_{lb} < |u_1|_{lb}\), by the induction hypothesis it suffices to show that \(\neg(u_2' \to^* u_1')\).
		However, if this were the case, then \(u_2 \to u_2' \to^* u_1' \to u_1\), since \(u_1' \in \scc(u_1)\).
		This contradicts our assumption that \(\neg(u_2 \to^* u_1)\), so it must be the case that \(\neg(u_2' \to^* u_1')\).
		Whence, \((u_1,u_2) \to (u_1', u_2') \to^* (w_1, w_2)\) for some \((w_1,w_2)\) satisfying (C1).\medskip

		Now assume that \(\scc(u_1) = \scc(u_2)\).
		Since \(X^\bullet\) is layered, \((\scc(u_1), \diredge)\) is acyclic.
		This has two consequences: First of all, it means that \(x \diredge y\) if and only if \(y \loopright x\) for any \(x,y \in \scc(u_1)\).
		Second, it means that \((\scc(u_1), \loopright^*)\) is a join semilattice.
		Therefore, for some \(v \in \scc(u_1)\), \(u_1 \loopright^* v\) and \(u_2 \loopright^* v\), and \(v \loopright^* v'\) for any \(v'\) satisfying \(u_1 \loopright^* v'\) and \(u_2 \loopright^* v'\).
		We consider three distinct cases: \(u_1 = v\), \(u_2 = v\), and \(u_1 \neq v \neq u_2\).

		In the first two cases, simply take \(w_1 = v\) and let \(w_2\) be the \(u_i\) that is not \(v\).
		Since \(u_1 \neq u_2\), this means either that \(u_1 \loopright^+ u_2\) or \(u_2 \loopright^+ u_1\).
		In either case, \((w_1,w_2)\) satisfies (C2).
		Since \(R\) is symmetric, \((w_1,w_2) \in R\).

		In the third case, \(u_i \loopright^+ v\) for \(i = 1,2\).
		We will show that \((u_1,u_2) \to^* (w_1,w_2) \in R\) for some \((w_1,w_2)\) satisfying either (C1), (C2), or (C3).
		A further case analysis is needed, made possible by the following claim.\medskip

		\noindent\textbf{Claim 1.} Where \(i =1,2\), there is a state \(v_i\) such that \(u_i \loopright^* v_i \loopright v\), and for any other \(v'\) such that \(v_i \loopright v'\), \(v \loopright^* v'\).

		\noindent\textit{Proof of claim 1.}
		Where \([u \loopright] = \{v' \mid u \loopright v'\}\), we begin by showing that \([u \loopright]\) is linearly ordered by \(\loopright\).
		To this end, suppose that \(v' \loopleft u \loopright v''\) for distinct \(v',v''\), even though neither \(v' \loopright v''\) nor \(v'' \loopright v'\).
		Since strongly connected components are preserved accross \(\loopright\), \(\scc(v') = \scc(u) = \scc(v'')\), so using the join semilattice structure of \((\scc(u), \loopright^*)\) we find a \(w \in \scc(u)\) such that \(v' \loopright^+ w \loopleft^+ v''\).
		Since \(\neg(v' \loopright v'')\), there is a body path \(u \bo^* w\) that does not pass through \(v''\).
		Hence, \(w \loopright^+ u \loopright v'' \loopright^+ w\), contradicting layeredness of \(X^\bullet\).
		It must have been the case that either \(v' \loopright v''\) or \(v''\loopright v'\).

		Now let \(u_i \loopright^* v' \loopright v\).
		Since \([v' \loopright]\) is finite and linear, either we can take \(v_i = v'\) or there is a \(\loopright\)-maximal state \(v_i\) of \([v'\loopright]\) such that \(v_i\loopright v\).
		In either case, \(u_i \loopright^* v_i \loopright v\) and \(v_i\) satisfies the desired condition.
		This concludes the proof of the claim.\medskip

		According to Claim 1, where \(i =1,2\), there is a state \(v_i\) such that \(u_i \loopright^* v_i \loopright v\), and for any other \(v'\) such that \(v_i \loopright v'\), \(v \loopright^* v'\).
		Furthermore, in choosing \(v\) to be the \(\loopright^*\)-join of \(u_1\) and \(u_2\), we assumed that \(v_1 \neq v_2\).
		Since \(X^\bullet\) is fully specified, either \(\neg(v_2 \bo^+ v_1)\) or \(\neg(v_1 \bo^+ v_2)\).
		Without loss of generalisation, assume \(\neg(v_2 \bo^+ v_1)\).

		At present, our assumptions are that 
		\begin{mathpar}
			u_1 \loopright^* v_1 \loopright v \loopleft v_2 \loopleft^* u_2
			\qquad
			(\forall v')~\text{if \(v_i \loopright v'\), then \(v' = v\) or \(v \loopright v'\)}
			\qquad
			\neg(v_2 \bo^+ v_1)
		\end{mathpar}
		and we claim that it follows from these assumptions that \((u_1,u_2) \to^* (w_1,w_2)\) for some \((w_1,w_2)\) satisfying either (C1), (C2), or (C3).
		We proceed, again, by induction on \(|u_1|_{lb}\).

		If \(|u_1|_{lb} = 0\), then \(u_1 = v_1\).
		Setting \(w_1 = u_1\) and \(w_2 = u_2\), we see that \(v_1\) witnesses the existential quantifier in (C3) for the pair \((w_1,w_2)\).
		That is, \((w_1,w_2)\) satisfies (C3).

		Otherwise, \(u_1 \bo^* u_1'\) for some \(u_1' \in \scc(u_1)\), with \(|u_1'|_{lb} < |u_1|_{lb}\).
		Since \((u_1,u_2) \in R\), \((u_1,u_2) \to (u_1',u_2') \in R\) for some \(u_2'\).
		If \(\scc(u_1') \neq \scc(u_2')\), then we return to the very first case in the proof of this lemma to conclude that \((u_1',u_2') \to^* (w_1,w_2)\) for some \((w_1,w_2)\) satisfying (C1).

		If \(\scc(u_1') = \scc(u_2')\), we have \(\scc(u_2) = \scc(u_1) = \scc(u_1') = \scc(u_2')\).
		Note that because \(u_1 \loopright^+ v_1\) and \(u_1 \bo u_1'\), we have \(u_1' \loopright^* v_1\) as well.
		The rest of the proof operates by breaking the statement \(u_2 \loopright^* v_2\) into its disjuncts.\medskip

		In the first case, let \(u_2 \loopright^+ v_2\).
		Since \(u_2 \to u_2'\), either \(u_2' = v_2\), \(u_2 \eo u_2'\), or \(u_2' \loopright^+ v_2\).
		In any case, \(u_2' \loopright^* v_2\), so we are back in the situation
		\begin{equation}\label{eq:induction hypothesis}
			u_1' \loopright^* v_1 \loopright v \loopleft v_2 \loopleft^* u_2'
			\qquad
			(\forall v')~\text{if \(v_i \loopright v'\), then \(v' = v\) or \(v \loopright v'\)}
			\qquad
			\neg(v_2 \bo^+ v_1)
		\end{equation}
		except with \(|u_1'|_{lb} < |u_1|_{lb}\).
		By the induction hypothesis, \((u_1',u_2') \to^* (w_1,w_2)\) for some \((w_1,w_2)\) satisfying one of (C1), (C2), or (C3).

		In the second case, let \(u_2 = v_2\).
		There are four further subcases to consider.
		\begin{itemize}
			\item[(i)] \(u_2 \eo u_2'\).
			In this subcase, \(u_2'\loopright^* u_2\), and we are back in the situation of (\ref{eq:induction hypothesis}) with \(|u_1'|_{lb} < |u_1|_{lb}\).
			If \(u_2' = u_1'\), then \(v_1 = v_2\) as well, contradicting \(\loopright^*\)-minimality of \(v\). 
			Hence, applying the induction hypothesis completes this subcase.
			\item[(ii)] \(u_2 \bo u_2'\).
			In this subcase, we obtain \(u_1' \neq u_2'\) from the observations that there is a body path \(u_1' \bo^* v_1\), that \(v_2 = u_2\), and we have assumed that \(\neg(v_2 \bo^+ v_1)\).

			To finish this final subcase, observe that there is still the possibility that \(v = u_2'\).
			In such a case, \(u_1' \loopright^* v_1 \loopright v = u_2'\), so setting \(w_1 = u_1'\) and \(w_2 = u_2'\) produces a pair \((w_1,w_2)\) satisfying (C3).

			Thus, the remaining situation to consider is the one in which \(v \neq u_2'\).
			In this case, we reconstruct the situation (\ref{eq:induction hypothesis}), but with fresh states taking the places of \(v_1\) and \(v_2\).
			Toward this end, observe that \(u_2'\loopright^+ v\), so that by claim 1 there exists a state \(v_2'\) such that \(u_2' \loopright^* v_2' \loopright v\) and for any \(v' \loopleft v_2'\) we have \(v \loopright v'\).
			Since \(u_2 = v_2 \bo u_2'\) and \(u_2' \loopright^* v_2'\), we see that \(\neg(v_2' \bo^+ v_1)\) follows from \(X^\bullet\) being goto-free.
			Hence, where \(v_1' = v_1\), we have
			\begin{mathpar}
			u_1' \loopright^* v_1' \loopright v \loopleft v_2' \loopleft^* u_2'
			\qquad
			(\forall v')~\text{if \(v_i' \loopright v'\), then \(v' = v\) or \(v \loopright v'\)}
			\qquad
			\neg(v_2' \bo^+ v_1')
		\end{mathpar}
		and \(|u_1'|_{lb} < |u_1|_{lb}\).
		By the induction hypothesis, \((u_1,u_2) \to^* (w_1,w_2)\) for some pair \((w_1,w_2)\) satisfying one of (C1), (C2), or (C3).
		\end{itemize}
		This exhausting case analysis concludes the proof.
	\end{proof}

	\begin{lemma}[\cite{grabmayerfokkink2020complete}]
		Let \(X\) be a well-layered prechart with layering witness \(X^\bullet\), \(R\) be a bisimulation equivalence on \(X\), and \((w_1,w_2) \in R\) be a pair satisfying one of (C1), (C2), or (C3) in \(X^\bullet\).
		Then \(X[w_2/w_1]\) is well-layered.
	\end{lemma}

	\begin{proof}(Sketch)
		Here, we simply mention the prechart versions of the entry/body labellings that appear in Grabmayer and Fokkink's proof of Proposition 6.8 from \cite{grabmayerfokkink2020complete}.
		Like for Lemma~\ref{lem:arbitrary bisimulations}, a nearly identical proof is sufficient, with only minor tweaks.

		There are three cases considered in Grabmayer and Fokkink's proof, corresponding to whether \((w_1,w_2)\) satisfies (C1), (C2), or (C3).
		\begin{itemize}
			\item[(C1)] Let \(Y^\bullet\) be the entry/body labelling of \(X[w_2/w_1]\) obtained from the rerouting \(X^\bullet[w_2/w_1]\) by replacing each \(v \eo v'\) such that \(\neg(v' \to^+ v)\) with a body transition \(v \bo v'\).
			\item[(C2)] Since \(w_2 \loopright^+ w_1\), there is a \(w_2'\) such that \(w_2 \loopright^* w_2' \loopright v\) and for any \(v' \loopleft w_2'\) we find \(v \loopright v'\).
			Let \(Y^\bullet\) be the entry/body labelling of \(X[w_2/w_1]\) obtained from the rerouting \(X^\bullet[w_2/w_1]\) by replacing each \(w_2' \bo v'\) with an entry transition \(w_2' \eo v'\), and then replacing each \(v \eo v'\) such that \(\neg(v' \to^+ v)\) with a body transition \(v \bo v'\).
			\item[(C3)] If \(v\) witnesses the quantifier in (C3), then \(v \neq w_1\).
			This means that \(v \in X[w_2/w_1]\).
			Let \(Y^\bullet\) be the entry/body labelling from the (C1) case.
		\end{itemize}
		In each case, \(Y^\bullet\) is a layering witness for \(X[w_2/w_1]\).
	\end{proof}

	\subsection{Generalized reroutings}
	Let \(X\) be a prechart, \((x_1, x_2)\) be a pair of states of \(X\), and define \(X[x_2/x_1]\) and \(X[i,j]\) as they were described in \autoref{sec:layered loop existence and elimination}, where \(i : X\setminus\{x_1\} \into X\) and \(j(x_1) = j(x_2)\).

	\begin{proposition}
		The connect-\(x_1\)-through-to-\(x_2\) construction coincides with the rerouting of \(X\) by \((i,j)\), ie. \(X[x_2/x_1] = X[i,j]\).
	\end{proposition}

	\begin{proof}
		Given \(y \neq x_1\) in \(X\), \(a \in A\),
		\[
			P(j) \circ \partial \circ i (y) (a)
			= P(j) \circ \partial(y)(a)
			= j(\partial(y)(a))
			= \begin{cases}
				\{x_2\} \cup (\partial(y)(a) \setminus \{x_1\}) &\text{if \(x_1 \in \partial(y)(a)\)}\\
				\partial(y)(a) &\text{otherwise}.
			\end{cases}
		\]
		This is precisely the definition of \(\partial[x_2/x_1](y)(a)\).
	\end{proof}

	\restatebisimulationreroutinglemma*

	\begin{proof}
		Let \((R, \delta_R)\) be the coalgebra structure on \(R\), and define \(i_Q : Q \into R\) and \(j_Q : R \onto Q\) to be the maps \(i_Q(x,y) = (x, i(y))\) and \(j_Q(x, z) = (x, j(z))\).
		We need to check that \(j_Q\) is, indeed, a map into \(Q\).
		This follows from the observation that \(j(z) = j\circ i\circ j(z)\), and therefore since \(\ker(j) \subseteq R\), \((z, i \circ j(z)) \in R\).
		Because \(R\) is transitive and \((x, z), (z, i \circ j(z)) \in R\), \((x, i\circ j(z)) \in R\) as well. 
		This means \(j_Q(x, z) = (x, j(z)) \in Q\).

		Define the coalgebra structure \((Q, \delta_Q[i_Q,j_Q])\), where
		\[
			\delta_Q[i_Q,j_Q] = G(j_Q) \circ \delta_R \circ i_Q.
		\] 
		By definition, \(i_Q\) and \(j_Q\) satisfy \(\pi_1^R \circ i_Q = \pi_1^Q\), \(\pi_2^R \circ i_Q = i \circ \pi_2^Q\), and \(\pi_2^R \circ j_Q = j \circ \pi_2^Q\).
		On the one hand, \(\pi_1^Q : Q \to X\) is a coalgebra homomorphism by definition.
		On the other, 
		\begin{align*}
			\delta[i,j] \circ \pi_2^Q &= G(j) \circ \delta_X \circ i \circ \pi_2^Q \tag{def. of \(\delta[i,j]\)}\\
			&= G(j) \circ \delta_X \circ \pi_2^R \circ i_Q \tag{def. of \(i_Q\)} \\
			&= G(j) \circ G(\pi_2^R) \circ \delta_R \circ i_Q \tag{\(R\) is a bisim.}\\
			&= G(j \circ \pi_2^R) \circ \delta_R \circ i_Q \tag{\(G\) is a functor}\\
			&= G(\pi_2^Q \circ j_Q) \circ \delta_R \circ i_Q \tag{def. of \(j_Q\)}\\
			&= G(\pi_2^Q) \circ G(j_Q) \circ \delta_R \circ i_Q \tag{\(G\) is a functor}\\
			&= G(\pi_2^Q) \circ \delta[i_Q,j_Q]. \tag{def. of \(\delta[i_Q,j_Q]\)}
		\end{align*}
		Thus, \(Q\) is a bisimulation between \(X\) and \(X[i,j]\).
	\end{proof}

	\restatefinitarybisimulationreroutingtheorem*
	
	\begin{proof}
		As (ii) follows from (i), it suffices to show (i).
		Let \(R\) be a bisimulation equivalence on a finite \(G\)-coalgebra \(X \in \mathcal C\).
		We proceed by induction on the number \(n = |R - \Delta_X|\), with the base case being vacuous.
		
		Let \((i,j)\) be a split pair, such that \(U = X[i,j] \in \mathcal C\) is a nontrivial rerouting of \(X\).
		Then, where \(Q = R \cap (U \times U)\), \(|Q-\Delta_U| < n\), as \(j\) identifies some pair of distinct states in \(R\). 
		By Lemma~\ref{lem:bisimulation rerouting lemma}, \(Q\) is a bisimulation equivalence on \(U\), and \(U \in \mathcal C\), so the induction hypothesis tells us \(U/Q \in \mathcal C\).
		Whence, it suffices to show that the composition \([-]_Q \circ j : X \to U \to U/Q\) has precisely \(R\) as its kernel.
		
		Clearly, \(\ker([-]_Q \circ j) \subseteq \ker(j) \subseteq R\).
		To see the converse, let \((x,y) \in R\).
		It suffices to show that \((j(x), j(y)) \in R\), since then \((j(x),j(y)) \in Q\) and hence \([j(x)]_Q = [j(y)]_Q\).
		To this end, observe that \(j\circ i \circ j(x) = j(x)\), so that \((i \circ j(x), x) \in R\) by assumption.
		Similarly, \(j(y) = j\circ i \circ j(y)\), so that \((y, i\circ j(y)) \in R\).
		Since \(R\) is transitive and \((x,y) \in R\), \((i \circ j(x), i \circ j(y)) \in R\).
		But \(i = \text{in}_U : U \into X\), telling us that \((j(x), j(y)) \in R\) as desired.
	\end{proof}

\section{Proofs from \autoref{sec:a global approach}}
	
\restateclosureunderarbitraryhomomorphicimages*

	\begin{proof}
		Since \(X\) locally finite, and local finiteness is preserved under homomorphic images, \(Y\) is locally finite as well.
		Let \(U\) be a finite subcoalgebra of \(Y\).
		By Theorem \ref{thm:closure under homomorphic images}, it suffices to show that \(U\) is the homomorphic image of a well-layered prechart.

		To this end, let \(U = \{y_1, \dots, y_n\}\) and let \(\{x_1, \dots, x_n\} \subseteq X\) be such that \(q(x_i) = y_i\) for \(i = 1,\dots n\).
		If \(V\) is the smallest subcoalgebra containing \(\{x_1, \dots, x_n\}\), then \(q\) restricts to a subcoalgebra homomorphism on \(V\), and by definition \(U\) is the smallest subcoalgebra of \(Y\) containing \(\{y_1,\dots,y_n\}\).
		Whence, the restriction \(q|_V : V \onto U\) is a surjective homomorphism.

		Since \(\langle x_i\rangle\) is finite for each \(i\), \(V\) is finite and therefore well-layered.
		Hence, \(U\) the homomorphic image of a finite well-layered prechart, and by Theorem \ref{thm:closure under homomorphic images} is well-layered.
	\end{proof}

	Let \(G\) be an endofunctor on \(\Sets\) that preserves weak pullbacks.
	Where \(E\) is a locally finite \(G\)-coalgebra and \(\equiv\) is a bisimulation equivalence on \(E\), assume that in the four steps of the local approach we have obtained a class \(\mathcal C\) of finite \(G\)-coalgebras such that
	\begin{itemize}
		\item[(a)] each \(X \in \mathcal C\) admits a unique homomorphism into \(E/\equiv\), 
		\item[(b)] \(\langle e \rangle \in \mathcal C\) for any \(e \in E\), and 
		\item[(c)] \(\mathcal C\) is closed under binary coproducts and homomorphic images.
	\end{itemize}
	Then the class \(\mathcal C_{loc}\) of \emph{locally \(\mathcal C\) coalgebras}, locally finite coalgebras \(X\) such that every finite subcoalgebra of \(X\) is in \(\mathcal C\), satisfies the necessary conditions for steps 2 through 4 of the global approach.
	
	\restatelocaltoglobal*
	
	We prove this in three parts.
	
	\begin{lemma}\label{lem:C_loc is closed under}
		\(\mathcal C_{loc}\) is closed under coproducts and homomorphic images.
	\end{lemma}

	\begin{proof}
		Let \(X,Y \in \mathcal C_{loc}\).
		Since \(X \sqcup Y\) is locally finite, it suffices to check that every finite subcoalgebra of \(X \sqcup Y\) is in \(\mathcal C\).
		To this end, let \(U\) be a finite subcoalgebra of \(X \sqcup Y\), and define \(U_1 = U \cap X\) and \(U_2 = U \cap Y\).
		Since \(G\) preserves weak pullbacks, intersections of subcoalgebras of \(X \sqcup Y\) are subcoalgebras, so \(U_1\) is a subcoalgebra of \(X\) and \(U_2\) is a subcoalgebra of \(Y\).
		Since \(U_1\) and \(U_2\) are subsets of \(U\), they are finite.
		This puts \(U_1,U_2 \in \mathcal C\) by assumption.
		Since \(\mathcal C\) is closed under binary coproducts, \(U = U_1 \sqcup U_2 \in \mathcal C\).
		\(U\) was arbitrary, so \(X \sqcup Y \in \mathcal C_{loc}\).
		
		Now suppose only \(X \in \mathcal C_{loc}\), but that \(q : X \onto Y\) is a surjective coalgebra homomorphism.
		Since \(X\) is locally finite, and local finiteness is preserved by surjective homomorphisms, \(Y\) is locally finite.
		Let \(V = \{y_1,\cdots,y_m\}\) be a finite subcoalgebra of \(Y\).
		Choose an \(x_i \in q^{-1}(y_i)\) for each \(i \le m\), and let \(U = \bigcup \langle x_i\rangle\).
		\(U\) is a finite union of finite sets, and subcoalgebras are preserved under unions in general, so \(U\) is a finite subcoalgebra of \(X\).
		This puts \(U \in \mathcal C\).
		Since \(\mathcal C\) is closed under homomorphic images, it suffices to show that \(q(U) = V\).
		
		We begin by observing that \(q(\langle x_i \rangle) \supseteq \langle y_i \rangle\) for each \(i \le m\) by definition, as \(q(\langle x_i \rangle)\) is a subcoalgebra of \(Y\) containing \(y_i\).
		Conversely, \(q^{-1}(\langle y_i \rangle) \supseteq \langle x_i \rangle\), by definition, as \(q^{-1}(\langle y_i \rangle)\) is a subcoalgebra of \(X\) containing \(x_i\).
		Hence, \(\langle y_i \rangle = q(q^{-1}(\langle y_i \rangle)) \supseteq q(\langle x_i \rangle)\), and therefore \(q(\langle x_i \rangle) = \langle y_i \rangle\) for each \(i \le m\).
		Now, images preserve unions, so \[
			q(U) = q(\bigcup \langle x_i\rangle) = \bigcup q(\langle x_i\rangle) = \bigcup \langle y_i\rangle = V.
		\]
		It follows that \(V \in \mathcal C\), so \(Y \in \mathcal C_{loc}\). 
	\end{proof}

	\begin{lemma}
		Every \(X \in \mathcal C_{loc}\) admits a unique homomorphism \(X \to E/{\equiv}\).
	\end{lemma}

	\begin{proof}
		Let \(X \in \mathcal C_{loc}\).
		Since \(X\) is locally \(\mathcal C\), \(\langle x \rangle \in \mathcal C\) for every \(x \in X\).
		By assumption, every finite subcoalgebra \(U\) of \(X\) admits a unique solution \(s_U : U \to E/{\equiv}\).
		This allows us to define the relation \(R = \{(x, s_{U}(x)) \mid x \in U \subseteq X, \text{ \(U\) a finite subcoalgebra of \(X\)}\}\).
		We argue that \(R\) is the graph of a homomorphism into \(E/{\equiv}\).
		
		To see that it is the graph of a function, suppose \((x, s_{U}(x)), (x, s_V(x)) \in R\).
		Then \(U \cap V\) is a finite subcoalgebra of \(X\), putting \(U \cap V \in \mathcal C\).
		Therefore, there is a unique solution \(s_{U \cap V} : U \cap V \to E/{\equiv}\) to \(U \cap V\).
		By composing the inclusion homomorphism \(U \cap V \into U \xrightarrow{s_U} E/{\equiv}\), we see that \(s_{U}|_{U \cap V}\) is a solution to \(U \cap V\), and similarly \(s_V|_{U \cap V}\) is a solution.
		We know that \(x \in X\), so by uniqueness of solutions we have\[
			s_U(x) = s_U|_{U \cap V}(x) = s_{U \cap V}(x) = s_V|_{U \cap V}(x) = s_V(x).
		\]
		Hence, \(R\) is the graph of a function, call it \(s : X \to E/{\equiv}\).
		
		Finally, to see that \(s\) is a coalgebra homomorphism, let \(x \in U \subseteq X\), where \(U\) is a finite subcoalgebra of \(X\).
		Then \[
			\delta_{E/{\equiv}} \circ s(x) = \delta_{E/{\equiv}} \circ s_U(x) = G(s_U) \circ \delta_U(x) = G(s) \circ G(\text{in}_U) \circ \delta_U(x) = G(s) \circ \delta_X \circ \text{in}_U(x) = G(s) \circ \delta_X(x).
		\]
		Hence, \(\delta_{E/{\equiv}} \circ s(x) = G(s) \circ \delta_X\), making \(s\) a solution to \(X\).
	\end{proof}

	\begin{lemma}
		\(E/{\equiv} \in \mathcal C_{loc}\)
	\end{lemma}

	\begin{proof}
		This follows from Lemma~\ref{lem:C_loc is closed under}, as \([-]_{\equiv} : E \to E/{\equiv}\) is a surjective coalgebra homomorphism and \(E\) is locally of the form \(\langle e \rangle \in \mathcal C\) for \(e \in \mathcal C\) and \(\mathcal C\) is closed under finite unions. 
	\end{proof}

	It follows from the three lemmas above that \(E/{\equiv} \in \mathcal C_{loc}\), and that every \(X \in \mathcal C_{loc}\) admits a unique homomorphism into \(E/{\equiv}\).
	This is what it means for \(E/{\equiv}\) to be final in \(\mathcal C_{loc}\).

%% file: eptcs-main.bbl
\begin{thebibliography}{10}
\providecommand{\bibitemdeclare}[2]{}
\providecommand{\surnamestart}{}
\providecommand{\surnameend}{}
\providecommand{\urlprefix}{Available at }
\providecommand{\url}[1]{\texttt{#1}}
\providecommand{\href}[2]{\texttt{#2}}
\providecommand{\urlalt}[2]{\href{#1}{#2}}
\providecommand{\doi}[1]{doi:\urlalt{http://dx.doi.org/#1}{#1}}
\providecommand{\eprint}[1]{arXiv:\urlalt{https://arxiv.org/abs/#1}{#1}}
\providecommand{\bibinfo}[2]{#2}

\bibitemdeclare{article}{adamek2005introduction}
\bibitem{adamek2005introduction}
\bibinfo{author}{Jiří \surnamestart Adámek\surnameend}
  (\bibinfo{year}{2005}): \emph{\bibinfo{title}{Introduction to coalgebra.}}
\newblock {\sl \bibinfo{journal}{Theory and Applications of Categories
  [electronic only]}} \bibinfo{volume}{14}, pp. \bibinfo{pages}{157--199}.

\bibitemdeclare{inproceedings}{BaetenCorradiniGrabmayer2006}
\bibitem{BaetenCorradiniGrabmayer2006}
\bibinfo{author}{Jos C.~M. \surnamestart Baeten\surnameend},
  \bibinfo{author}{Flavio \surnamestart Corradini\surnameend} \&
  \bibinfo{author}{Clemens \surnamestart Grabmayer\surnameend}
  (\bibinfo{year}{2006}): \emph{\bibinfo{title}{On the Star Height of Regular
  Expressions Under Bisimulation (Extended Abstract)}}.
\newblock \bibinfo{series}{EXPRESS '06}.

\bibitemdeclare{article}{BaetenCG07}
\bibitem{BaetenCG07}
\bibinfo{author}{Jos C.~M. \surnamestart Baeten\surnameend},
  \bibinfo{author}{Flavio \surnamestart Corradini\surnameend} \&
  \bibinfo{author}{Clemens \surnamestart Grabmayer\surnameend}
  (\bibinfo{year}{2007}): \emph{\bibinfo{title}{A characterization of regular
  expressions under bisimulation}}.
\newblock {\sl \bibinfo{journal}{J. {ACM}}}
  \bibinfo{volume}{54}(\bibinfo{number}{2}), p.~\bibinfo{pages}{6},
  \doi{10.1145/1219092.1219094}.

\bibitemdeclare{book}{bang0jensengutin2009digraphs}
\bibitem{bang0jensengutin2009digraphs}
\bibinfo{author}{J{\o}rgen \surnamestart Bang{-}Jensen\surnameend} \&
  \bibinfo{author}{Gregory~Z. \surnamestart Gutin\surnameend}
  (\bibinfo{year}{2009}): \emph{\bibinfo{title}{Digraphs - Theory, Algorithms
  and Applications, Second Edition}}.
\newblock \bibinfo{series}{Springer Monographs in Mathematics},
  \bibinfo{publisher}{Springer}, \doi{10.1007/978-1-84800-998-1}.

\bibitemdeclare{article}{bergstrabethkeponse1994process}
\bibitem{bergstrabethkeponse1994process}
\bibinfo{author}{J.~\surnamestart Bergstra\surnameend},
  \bibinfo{author}{I.~\surnamestart Bethke\surnameend} \&
  \bibinfo{author}{A.~\surnamestart Ponse\surnameend} (\bibinfo{year}{1994}):
  \emph{\bibinfo{title}{Process Algebra with Iteration and Nesting}}.
\newblock {\sl \bibinfo{journal}{The Computer Journal}}
  \bibinfo{volume}{37}(\bibinfo{number}{4}), pp. \bibinfo{pages}{243--258},
  \doi{10.1093/comjnl/37.4.243}.
\newblock
  \eprint{https://academic.oup.com/comjnl/article-pdf/37/4/243/1067027/370243.pdf}.

\bibitemdeclare{article}{bonsanguemiliussilva2013sound}
\bibitem{bonsanguemiliussilva2013sound}
\bibinfo{author}{Marcello~M. \surnamestart Bonsangue\surnameend},
  \bibinfo{author}{Stefan \surnamestart Milius\surnameend} \&
  \bibinfo{author}{Alexandra \surnamestart Silva\surnameend}
  (\bibinfo{year}{2013}): \emph{\bibinfo{title}{Sound and Complete
  Axiomatizations of Coalgebraic Language Equivalence}}.
\newblock {\sl \bibinfo{journal}{ACM Trans. Comput. Logic}}
  \bibinfo{volume}{14}(\bibinfo{number}{1}), \doi{10.1145/2422085.2422092}.

\bibitemdeclare{article}{brzozowski1964derivatives}
\bibitem{brzozowski1964derivatives}
\bibinfo{author}{Janusz~A. \surnamestart Brzozowski\surnameend}
  (\bibinfo{year}{1964}): \emph{\bibinfo{title}{Derivatives of Regular
  Expressions}}.
\newblock {\sl \bibinfo{journal}{J. {ACM}}}
  \bibinfo{volume}{11}(\bibinfo{number}{4}), pp. \bibinfo{pages}{481--494},
  \doi{10.1145/321239.321249}.

\bibitemdeclare{inproceedings}{chenpucella2003coalgebraickat}
\bibitem{chenpucella2003coalgebraickat}
\bibinfo{author}{Hubie \surnamestart Chen\surnameend} \&
  \bibinfo{author}{Riccardo \surnamestart Pucella\surnameend}
  (\bibinfo{year}{2003}): \emph{\bibinfo{title}{A Coalgebraic Approach to
  Kleene Algebra with Tests}}.
\newblock In \bibinfo{editor}{H.~Peter \surnamestart Gumm\surnameend}, editor:
  {\sl \bibinfo{booktitle}{6th International Workshop on Coalgebraic Methods in
  Computer Science, {CMCS} 2003, Satellite Event for {ETAPS} 2003, Warsaw,
  Poland, April 5-6, 2003}}, {\sl \bibinfo{series}{Electronic Notes in
  Theoretical Computer Science}}~\bibinfo{volume}{82},
  \bibinfo{publisher}{Elsevier}, pp. \bibinfo{pages}{94--109},
  \doi{10.1016/S1571-0661(04)80634-0}.

\bibitemdeclare{book}{conway2012regular}
\bibitem{conway2012regular}
\bibinfo{author}{John~Horton \surnamestart Conway\surnameend}
  (\bibinfo{year}{2012}): \emph{\bibinfo{title}{Regular algebra and finite
  machines}}.
\newblock \bibinfo{publisher}{Courier Corporation}.

\bibitemdeclare{inproceedings}{fokkink1997perpetual}
\bibitem{fokkink1997perpetual}
\bibinfo{author}{Wan \surnamestart Fokkink\surnameend} (\bibinfo{year}{1997}):
  \emph{\bibinfo{title}{Axiomatizations for the perpetual loop in process
  algebra}}.
\newblock In \bibinfo{editor}{Pierpaolo \surnamestart Degano\surnameend},
  \bibinfo{editor}{Roberto \surnamestart Gorrieri\surnameend} \&
  \bibinfo{editor}{Alberto \surnamestart Marchetti-Spaccamela\surnameend},
  editors: {\sl \bibinfo{booktitle}{Automata, Languages and Programming}},
  \bibinfo{publisher}{Springer Berlin Heidelberg}, \bibinfo{address}{Berlin,
  Heidelberg}, pp. \bibinfo{pages}{571--581}, \doi{10.1145/321312.321326}.

\bibitemdeclare{article}{fokkinkzantema1994basic}
\bibitem{fokkinkzantema1994basic}
\bibinfo{author}{Wan~J. \surnamestart Fokkink\surnameend} \&
  \bibinfo{author}{Hans \surnamestart Zantema\surnameend}
  (\bibinfo{year}{1994}): \emph{\bibinfo{title}{Basic Process Algebra with
  Iteration: Completeness of its Equational Axioms}}.
\newblock {\sl \bibinfo{journal}{Comput. J.}}
  \bibinfo{volume}{37}(\bibinfo{number}{4}), pp. \bibinfo{pages}{259--268},
  \doi{10.1093/comjnl/37.4.259}.

\bibitemdeclare{article}{fokkinkzantema1997termination}
\bibitem{fokkinkzantema1997termination}
\bibinfo{author}{Wan~J. \surnamestart Fokkink\surnameend} \&
  \bibinfo{author}{Hans \surnamestart Zantema\surnameend}
  (\bibinfo{year}{1997}): \emph{\bibinfo{title}{Termination Modulo Equations by
  Abstract Commutation with an Application to Iteration}}.
\newblock {\sl \bibinfo{journal}{Theor. Comput. Sci.}}
  \bibinfo{volume}{177}(\bibinfo{number}{2}), pp. \bibinfo{pages}{407--423},
  \doi{10.1016/S0304-3975(96)00254-X}.

\bibitemdeclare{inproceedings}{fosterkozenmilanosilvathompson2015netkat}
\bibitem{fosterkozenmilanosilvathompson2015netkat}
\bibinfo{author}{Nate \surnamestart Foster\surnameend}, \bibinfo{author}{Dexter
  \surnamestart Kozen\surnameend}, \bibinfo{author}{Matthew \surnamestart
  Milano\surnameend}, \bibinfo{author}{Alexandra \surnamestart
  Silva\surnameend} \& \bibinfo{author}{Laure \surnamestart
  Thompson\surnameend} (\bibinfo{year}{2015}): \emph{\bibinfo{title}{A
  Coalgebraic Decision Procedure for NetKAT}}.
\newblock In \bibinfo{editor}{Sriram~K. \surnamestart Rajamani\surnameend} \&
  \bibinfo{editor}{David \surnamestart Walker\surnameend}, editors: {\sl
  \bibinfo{booktitle}{Proceedings of the 42nd Annual {ACM} {SIGPLAN-SIGACT}
  Symposium on Principles of Programming Languages, {POPL} 2015, Mumbai, India,
  January 15-17, 2015}}, \bibinfo{publisher}{{ACM}}, pp.
  \bibinfo{pages}{343--355}, \doi{10.1145/2676726.2677011}.

\bibitemdeclare{article}{grabmayerfokkink2020complete}
\bibitem{grabmayerfokkink2020complete}
\bibinfo{author}{Clemens \surnamestart Grabmayer\surnameend} \&
  \bibinfo{author}{Wan \surnamestart Fokkink\surnameend}
  (\bibinfo{year}{2020}): \emph{\bibinfo{title}{A Complete Proof System for
  1-Free Regular Expressions Modulo Bisimilarity}}.
\newblock {\sl \bibinfo{journal}{Proceedings of the 35th Annual ACM/IEEE
  Symposium on Logic in Computer Science}}, \doi{10.1145/3373718.3394744}.

\bibitemdeclare{misc}{grabmayerfokkink2020extended}
\bibitem{grabmayerfokkink2020extended}
\bibinfo{author}{Clemens \surnamestart Grabmayer\surnameend} \&
  \bibinfo{author}{Wan \surnamestart Fokkink\surnameend}
  (\bibinfo{year}{2020}): \emph{\bibinfo{title}{A Complete Proof System for
  1-Free Regular Expressions Modulo Bisimilarity}}.
\newblock \eprint{2004.12740}.

\bibitemdeclare{article}{gumm1998functorsforcoalgebras}
\bibitem{gumm1998functorsforcoalgebras}
\bibinfo{author}{H.~\surnamestart Gumm\surnameend} (\bibinfo{year}{1998}):
  \emph{\bibinfo{title}{Functors for Coalgebras}}.
\newblock {\sl \bibinfo{journal}{Algebra Universalis}} \bibinfo{volume}{45},
  \doi{10.1007/s00012-001-8156-x}.

\bibitemdeclare{inproceedings}{gumm1999elements}
\bibitem{gumm1999elements}
\bibinfo{author}{H.~\surnamestart Gumm\surnameend} (\bibinfo{year}{1999}):
  \emph{\bibinfo{title}{Elements Of The General Theory Of Coalgebras}}.

\bibitemdeclare{article}{gummschroeder2001covarieties}
\bibitem{gummschroeder2001covarieties}
\bibinfo{author}{H.~Peter \surnamestart Gumm\surnameend} \&
  \bibinfo{author}{Tobias \surnamestart Schr{\"{o}}der\surnameend}
  (\bibinfo{year}{2001}): \emph{\bibinfo{title}{Covarieties and complete
  covarieties}}.
\newblock {\sl \bibinfo{journal}{Theor. Comput. Sci.}}
  \bibinfo{volume}{260}(\bibinfo{number}{1-2}), pp. \bibinfo{pages}{71--86},
  \doi{10.1016/S0304-3975(00)00123-7}.

\bibitemdeclare{inproceedings}{jacobs2006bialgebraic}
\bibitem{jacobs2006bialgebraic}
\bibinfo{author}{Bart \surnamestart Jacobs\surnameend} (\bibinfo{year}{2006}):
  \emph{\bibinfo{title}{A Bialgebraic Review of Deterministic Automata, Regular
  Expressions and Languages}}.
\newblock In \bibinfo{editor}{Kokichi \surnamestart Futatsugi\surnameend},
  \bibinfo{editor}{Jean{-}Pierre \surnamestart Jouannaud\surnameend} \&
  \bibinfo{editor}{Jos{\'{e}} \surnamestart Meseguer\surnameend}, editors: {\sl
  \bibinfo{booktitle}{Algebra, Meaning, and Computation, Essays Dedicated to
  Joseph A. Goguen on the Occasion of His 65th Birthday}}, {\sl
  \bibinfo{series}{Lecture Notes in Computer Science}} \bibinfo{volume}{4060},
  \bibinfo{publisher}{Springer}, pp. \bibinfo{pages}{375--404},
  \doi{10.1007/11780274\_20}.

\bibitemdeclare{book}{jacobs2016introduction}
\bibitem{jacobs2016introduction}
\bibinfo{author}{Bart \surnamestart Jacobs\surnameend} (\bibinfo{year}{2016}):
  \emph{\bibinfo{title}{Introduction to Coalgebra: Towards Mathematics of
  States and Observation}}.
\newblock {\sl \bibinfo{series}{Cambridge Tracts in Theoretical Computer
  Science}}~\bibinfo{volume}{59}, \bibinfo{publisher}{Cambridge University
  Press}, \doi{10.1017/CBO9781316823187}.

\bibitemdeclare{inproceedings}{jipsen2014concurrent}
\bibitem{jipsen2014concurrent}
\bibinfo{author}{Peter \surnamestart Jipsen\surnameend} (\bibinfo{year}{2014}):
  \emph{\bibinfo{title}{Concurrent Kleene Algebra with Tests}}.
\newblock In \bibinfo{editor}{Peter \surnamestart H{\"{o}}fner\surnameend},
  \bibinfo{editor}{Peter \surnamestart Jipsen\surnameend},
  \bibinfo{editor}{Wolfram \surnamestart Kahl\surnameend} \&
  \bibinfo{editor}{Martin~Eric \surnamestart M{\"{u}}ller\surnameend}, editors:
  {\sl \bibinfo{booktitle}{Relational and Algebraic Methods in Computer Science
  - 14th International Conference, RAMiCS 2014, Marienstatt, Germany, April
  28-May 1, 2014. Proceedings}}, {\sl \bibinfo{series}{Lecture Notes in
  Computer Science}} \bibinfo{volume}{8428}, \bibinfo{publisher}{Springer}, pp.
  \bibinfo{pages}{37--48}, \doi{10.1007/978-3-319-06251-8\_3}.

\bibitemdeclare{article}{johnstonepowertsujishitawatanabeworrell2001structure}
\bibitem{johnstonepowertsujishitawatanabeworrell2001structure}
\bibinfo{author}{Peter~T. \surnamestart Johnstone\surnameend},
  \bibinfo{author}{John \surnamestart Power\surnameend}, \bibinfo{author}{Toru
  \surnamestart Tsujishita\surnameend}, \bibinfo{author}{Hiroshi \surnamestart
  Watanabe\surnameend} \& \bibinfo{author}{James \surnamestart
  Worrell\surnameend} (\bibinfo{year}{2001}): \emph{\bibinfo{title}{On the
  structure of categories of coalgebras}}.
\newblock {\sl \bibinfo{journal}{Theor. Comput. Sci.}}
  \bibinfo{volume}{260}(\bibinfo{number}{1-2}), pp. \bibinfo{pages}{87--117},
  \doi{10.1016/S0304-3975(00)00124-9}.

\bibitemdeclare{inproceedings}{kappebrunetsilvazanasi2017concurrent}
\bibitem{kappebrunetsilvazanasi2017concurrent}
\bibinfo{author}{Tobias \surnamestart Kapp{\'e}\surnameend},
  \bibinfo{author}{Paul \surnamestart Brunet\surnameend},
  \bibinfo{author}{Alexandra \surnamestart Silva\surnameend} \&
  \bibinfo{author}{Fabio \surnamestart Zanasi\surnameend}
  (\bibinfo{year}{2018}): \emph{\bibinfo{title}{Concurrent Kleene Algebra: Free
  Model and Completeness}}.
\newblock In \bibinfo{editor}{Amal \surnamestart Ahmed\surnameend}, editor:
  {\sl \bibinfo{booktitle}{Programming Languages and Systems}},
  \bibinfo{publisher}{Springer International Publishing},
  \bibinfo{address}{Cham}, pp. \bibinfo{pages}{856--882},
  \doi{10.1007/978-3-319-89884-1_30}.

\bibitemdeclare{inproceedings}{kleene1951representation}
\bibitem{kleene1951representation}
\bibinfo{author}{S.~\surnamestart Kleene\surnameend} (\bibinfo{year}{1951}):
  \emph{\bibinfo{title}{Representation of Events in Nerve Nets and Finite
  Automata}}.

\bibitemdeclare{inproceedings}{kozen1991completeness}
\bibitem{kozen1991completeness}
\bibinfo{author}{Dexter \surnamestart Kozen\surnameend} (\bibinfo{year}{1991}):
  \emph{\bibinfo{title}{A Completeness Theorem for Kleene Algebras and the
  Algebra of Regular Events}}.
\newblock In: {\sl \bibinfo{booktitle}{Proceedings of the Sixth Annual
  Symposium on Logic in Computer Science {(LICS} '91), Amsterdam, The
  Netherlands, July 15-18, 1991}}, \bibinfo{publisher}{{IEEE} Computer
  Society}, pp. \bibinfo{pages}{214--225}, \doi{10.1109/LICS.1991.151646}.

\bibitemdeclare{inproceedings}{kozensmith1996kat}
\bibitem{kozensmith1996kat}
\bibinfo{author}{Dexter \surnamestart Kozen\surnameend} \&
  \bibinfo{author}{Frederick \surnamestart Smith\surnameend}
  (\bibinfo{year}{1996}): \emph{\bibinfo{title}{Kleene Algebra with Tests:
  Completeness and Decidability}}.
\newblock In \bibinfo{editor}{Dirk \surnamestart van Dalen\surnameend} \&
  \bibinfo{editor}{Marc \surnamestart Bezem\surnameend}, editors: {\sl
  \bibinfo{booktitle}{Computer Science Logic, 10th International Workshop,
  {CSL} '96, Annual Conference of the EACSL, Utrecht, The Netherlands,
  September 21-27, 1996, Selected Papers}}, {\sl \bibinfo{series}{Lecture Notes
  in Computer Science}} \bibinfo{volume}{1258}, \bibinfo{publisher}{Springer},
  pp. \bibinfo{pages}{244--259}, \doi{10.1007/3-540-63172-0\_43}.

\bibitemdeclare{inproceedings}{milius2010streamcircuits}
\bibitem{milius2010streamcircuits}
\bibinfo{author}{Stefan \surnamestart Milius\surnameend}
  (\bibinfo{year}{2010}): \emph{\bibinfo{title}{A Sound and Complete Calculus
  for Finite Stream Circuits}}.
\newblock In: {\sl \bibinfo{booktitle}{Proceedings of the 25th Annual {IEEE}
  Symposium on Logic in Computer Science, {LICS} 2010, 11-14 July 2010,
  Edinburgh, United Kingdom}}, \bibinfo{publisher}{{IEEE} Computer Society},
  pp. \bibinfo{pages}{421--430}, \doi{10.1109/LICS.2010.11}.

\bibitemdeclare{article}{milner1984complete}
\bibitem{milner1984complete}
\bibinfo{author}{Robin \surnamestart Milner\surnameend} (\bibinfo{year}{1984}):
  \emph{\bibinfo{title}{A Complete Inference System for a Class of Regular
  Behaviours}}.
\newblock {\sl \bibinfo{journal}{J. Comput. Syst. Sci.}}
  \bibinfo{volume}{28}(\bibinfo{number}{3}), pp. \bibinfo{pages}{439--466},
  \doi{10.1016/0022-0000(84)90023-0}.

\bibitemdeclare{inproceedings}{rutten1998coinduction}
\bibitem{rutten1998coinduction}
\bibinfo{author}{Jan J. M.~M. \surnamestart Rutten\surnameend}
  (\bibinfo{year}{1998}): \emph{\bibinfo{title}{Automata and Coinduction (An
  Exercise in Coalgebra)}}.
\newblock In \bibinfo{editor}{Davide \surnamestart Sangiorgi\surnameend} \&
  \bibinfo{editor}{Robert \surnamestart de~Simone\surnameend}, editors: {\sl
  \bibinfo{booktitle}{{CONCUR} '98: Concurrency Theory, 9th International
  Conference, Nice, France, September 8-11, 1998, Proceedings}}, {\sl
  \bibinfo{series}{Lecture Notes in Computer Science}} \bibinfo{volume}{1466},
  \bibinfo{publisher}{Springer}, pp. \bibinfo{pages}{194--218},
  \doi{10.1007/BFb0055624}.

\bibitemdeclare{article}{rutten2000universal}
\bibitem{rutten2000universal}
\bibinfo{author}{Jan J. M.~M. \surnamestart Rutten\surnameend}
  (\bibinfo{year}{2000}): \emph{\bibinfo{title}{Universal coalgebra: a theory
  of systems}}.
\newblock {\sl \bibinfo{journal}{Theor. Comput. Sci.}}
  \bibinfo{volume}{249}(\bibinfo{number}{1}), pp. \bibinfo{pages}{3--80},
  \doi{10.1016/S0304-3975(00)00056-6}.

\bibitemdeclare{article}{salomaa1966two}
\bibitem{salomaa1966two}
\bibinfo{author}{Arto \surnamestart Salomaa\surnameend} (\bibinfo{year}{1966}):
  \emph{\bibinfo{title}{Two Complete Axiom Systems for the Algebra of Regular
  Events}}.
\newblock {\sl \bibinfo{journal}{J. {ACM}}}
  \bibinfo{volume}{13}(\bibinfo{number}{1}), pp. \bibinfo{pages}{158--169},
  \doi{10.1145/321312.321326}.

\bibitemdeclare{inproceedings}{schmidkappekozensilva2021gkat}
\bibitem{schmidkappekozensilva2021gkat}
\bibinfo{author}{Todd \surnamestart Schmid\surnameend}, \bibinfo{author}{Tobias
  \surnamestart Kapp\'{e}\surnameend}, \bibinfo{author}{Dexter \surnamestart
  Kozen\surnameend} \& \bibinfo{author}{Alexandra \surnamestart
  Silva\surnameend} (\bibinfo{year}{2021}): \emph{\bibinfo{title}{{Guarded
  Kleene Algebra with Tests: Coequations, Coinduction, and Completeness}}}.
\newblock In \bibinfo{editor}{Nikhil \surnamestart Bansal\surnameend},
  \bibinfo{editor}{Emanuela \surnamestart Merelli\surnameend} \&
  \bibinfo{editor}{James \surnamestart Worrell\surnameend}, editors: {\sl
  \bibinfo{booktitle}{48th International Colloquium on Automata, Languages, and
  Programming (ICALP 2021)}}, {\sl \bibinfo{series}{Leibniz International
  Proceedings in Informatics (LIPIcs)}} \bibinfo{volume}{198},
  \bibinfo{publisher}{Schloss Dagstuhl -- Leibniz-Zentrum f{\"u}r Informatik},
  \bibinfo{address}{Dagstuhl, Germany}, pp. \bibinfo{pages}{142:1--142:14},
  \doi{10.4230/LIPIcs.ICALP.2021.142}.

\bibitemdeclare{misc}{schmid2021star}
\bibitem{schmid2021star}
\bibinfo{author}{Todd \surnamestart Schmid\surnameend},
  \bibinfo{author}{Jurriaan \surnamestart Rot\surnameend} \&
  \bibinfo{author}{Alexandra \surnamestart Silva\surnameend}
  (\bibinfo{year}{2021}): \emph{\bibinfo{title}{On Star Expressions and
  Coalgebraic Completeness Theorems}}.
\newblock \eprint{2106.08074}.

\bibitemdeclare{phdthesis}{silva2010kleene}
\bibitem{silva2010kleene}
\bibinfo{author}{Alexandra \surnamestart Silva\surnameend}
  (\bibinfo{year}{2010}): \emph{\bibinfo{title}{Kleene coalgebra}}.
\newblock Ph.D. thesis, \bibinfo{school}{University of Nijmegen}.

\bibitemdeclare{article}{silvabonsanguerutten2010nondeterministic}
\bibitem{silvabonsanguerutten2010nondeterministic}
\bibinfo{author}{Alexandra \surnamestart Silva\surnameend},
  \bibinfo{author}{Marcello \surnamestart Bonsangue\surnameend} \&
  \bibinfo{author}{Jan \surnamestart Rutten\surnameend} (\bibinfo{year}{2010}):
  \emph{\bibinfo{title}{Non-Deterministic Kleene Coalgebras}}.
\newblock {\sl \bibinfo{journal}{Logical Methods in Computer Science}}
  \bibinfo{volume}{6}(\bibinfo{number}{3}), \doi{10.2168/lmcs-6(3:23)2010}.

\bibitemdeclare{article}{smolkafosterhsukappekozensilva2019gkat}
\bibitem{smolkafosterhsukappekozensilva2019gkat}
\bibinfo{author}{Steffen \surnamestart Smolka\surnameend},
  \bibinfo{author}{Nate \surnamestart Foster\surnameend},
  \bibinfo{author}{Justin \surnamestart Hsu\surnameend},
  \bibinfo{author}{Tobias \surnamestart Kapp\'{e}\surnameend},
  \bibinfo{author}{Dexter \surnamestart Kozen\surnameend} \&
  \bibinfo{author}{Alexandra \surnamestart Silva\surnameend}
  (\bibinfo{year}{2019}): \emph{\bibinfo{title}{Guarded Kleene Algebra with
  Tests: Verification of Uninterpreted Programs in Nearly Linear Time}}.
\newblock {\sl \bibinfo{journal}{Proc. ACM Program. Lang.}}
  \bibinfo{volume}{4}(\bibinfo{number}{POPL}), \doi{10.1145/3371129}.

\bibitemdeclare{inproceedings}{turiplotkin1997operational}
\bibitem{turiplotkin1997operational}
\bibinfo{author}{Daniele \surnamestart Turi\surnameend} \&
  \bibinfo{author}{Gordon~D. \surnamestart Plotkin\surnameend}
  (\bibinfo{year}{1997}): \emph{\bibinfo{title}{Towards a Mathematical
  Operational Semantics}}.
\newblock In: {\sl \bibinfo{booktitle}{Proceedings, 12th Annual {IEEE}
  Symposium on Logic in Computer Science, Warsaw, Poland, June 29 - July 2,
  1997}}, \bibinfo{publisher}{{IEEE} Computer Society}, pp.
  \bibinfo{pages}{280--291}, \doi{10.1109/LICS.1997.614955}.

\end{thebibliography}
